\documentclass{article}
\usepackage{lmodern}
\usepackage[T1]{fontenc}

\usepackage{thanks}
\usepackage{latexsym}
\usepackage{amssymb}
\usepackage{amsthm}
\usepackage{amsmath}

\usepackage{amstext}
\usepackage{amsopn}

\usepackage{enumerate}
\usepackage{graphics}
\usepackage[all]{xy}
\usepackage{url}
\usepackage{authblk}

\newtheorem{theorem}{Theorem}

\newtheorem{lemma}[theorem]{Lemma}
\newtheorem {example}[theorem]{Example}

\newtheorem{note}[theorem]{Note}
\newtheorem{property}{Property}

\newcommand{\sgraph}{G}

\newcommand{\weight}{w}
\newcommand{\neighbour}{N}
\newcommand{\products}{\mathcal{P}}

\newcommand{\snet}{\mathcal{S}}
\newcommand{\prodset}{P}

\newcommand{\obar}[1]{\overline{#1}}

\newcommand{\srcnodes}{\mathit{source}}
\newcommand{\payoff}{p}
\newcommand{\strprofile}{s}

\newcommand{\agents}{\mathcal{A}}
\newcommand{\spred}{\operatorname{Pred}}
\newcommand{\nat}{\mathbb{N}}

\newcommand{\inflset}{\mathcal{N}}

\newcommand{\dagrank}{\operatorname{rank}}

\newcommand{\argmax}{\operatornamewithlimits{argmax}}

\newcommand{\constutil}{c_0}
\newcommand{\setn}[1]{[#1]}
\newcommand{\val}{\operatorname{Val}}
\newcommand{\bigo}{O}

\newcommand{\ES}{\emptyset}

\newcounter{symbol}
\newcommand{\indexsyma}[1]%
{\stepcounter{symbol}\index{zzz1 \thesymbol @\protect#1}}
\newcommand{\indexsymb}[1]%
{\stepcounter{symbol}\index{zzz2 \thesymbol @\protect#1}}
\newcommand{\indexsymc}[1]%
{\stepcounter{symbol}\index{zzz3 \thesymbol @\protect#1}}
\newcommand{\indexsymd}[1]%
{\stepcounter{symbol}\index{zzz4 \thesymbol @\protect#1}}
\newcommand{\indexsyme}[1]%
{\stepcounter{symbol}\index{zzz5 \thesymbol @\protect#1}}

\newcommand{\bfe}[1]{\begin{bfseries}\emph{#1}\end{bfseries}\index{#1}}
\newcommand{\oldbfe}[1]{\begin{bfseries}\emph{#1}\end{bfseries}}

\newcommand{\La}{\mbox{$\:\Leftarrow\:$}}
\newcommand{\Ra}{\mbox{$\:\Rightarrow\:$}}

\newcommand{\sse}{\mbox{$\:\subseteq\:$}}

\newcommand{\fa}{\mbox{$\forall$}}

\newcommand{\LLn}{\mbox{$1,\ldots,n$}}
\newcommand{\LL}{\mbox{$\ldots$}}

\newcommand{\C}[1]{\mbox{$\{{#1}\}$}}           % curly braces

\newcommand{\NI}{\noindent}
\newcommand{\HB}{\hfill{$\Box$}}
\newcommand{\VV}{\vspace{5 mm}}
\newcommand{\III}{\vspace{3 mm}}
\newcommand{\II}{\vspace{2 mm}}

\newcommand{\szkew}[1]{\relax \setbox0=\hbox{\kern -24pt $\displaystyle#1$\kern 0pt }%
\box0}
{\catcode`\@=11 \global\let\ifjusthvtest@=\iffalse}
\newcounter{oldmycaption}

\newcommand{\ceiling}[1]{\lceil #1 \rceil}

\title{Social Network Games}

\author{Sunil Simon\\
      Centre for Mathematics and Computer Science (CWI), Amsterdam\\
E-mail: \texttt{s.e.simon@cwi.nl}
  \and 
Krzysztof R. Apt\\
      Centre for Mathematics and Computer Science (CWI)
      and ILLC, University of Amsterdam, The Netherlands\\
E-mail: \texttt{k.r.apt@cwi.nl}
    }

\date{}

\begin{document}

\maketitle

\begin{abstract}
  One of the natural objectives of the field of the social networks is
  to predict agents' behaviour. To better understand the spread of
  various products through a social network \cite{AM11} introduced a
  threshold model, in which the nodes influenced by their neighbours
  can adopt one out of several alternatives. To analyze the
  consequences of such product adoption we associate here with each such
  social network a natural strategic game between the agents.
  
  In these games the payoff of each player weakly increases when more
  players choose his strategy, which is exactly opposite to the
  congestion games. The possibility of not choosing any product
  results in two special types of (pure) Nash equilibria.
  
  We show that such games may have no Nash equilibrium and that
  determining an existence of a Nash equilibrium, also of a special
  type, is NP-complete. This implies the same result for a more
  general class of games, namely polymatrix games.  The situation
  changes when the underlying graph of the social network is a DAG, a
  simple cycle, or, more generally, has no source nodes. For these
  three classes we determine the complexity of an existence of (a
  special type of) Nash equilibria.
  
  We also clarify for these categories of games the status and the
  complexity of the finite best response property (FBRP) and the
  finite improvement property (FIP). Further, we introduce a new
  property of the uniform FIP which is satisfied when the underlying
  graph is a simple cycle, but determining it is co-NP-hard in the
  general case and also when the underlying graph has no source
  nodes. The latter complexity results also hold for the
  property of being a weakly acyclic game. A preliminary version of this paper appeared as \cite{SA12}.

\end{abstract}

\section{Introduction}

\subsection{Background}

Social networks are a thriving interdisciplinary research area with
links to sociology, economics, epidemiology, computer
science, and mathematics.  A flurry of numerous articles and recent
books, see, e.g., \cite{EK10}, testifies to the
relevance of this field. It deals with such diverse topics as
epidemics, analysis of the connectivity, 
spread of certain patterns of social behaviour, effects of
advertising, and emergence of `bubbles' in financial markets.

One of the prevalent types of models of social networks are the {\em
  threshold models} introduced in~\cite{Gra78} and~\cite{Sch78}.  In
such a setup each node $i$ has a threshold $\theta(i) \in (0,1]$ and
  adopts an `item' given in advance (which can be a disease, trend, or
  a specific product) when the total weight of incoming edges from the
  nodes that have already adopted this item exceeds the threshold.
  One of the most important issues studied in the threshold models has
  been that of the spread of an item, see, e.g.,
  \cite{Mor00,KKT03,Che09}.  From now on we shall refer to an `item'
  that is spread by a more specific name of a `product'.

In this context very few papers dealt with more than one product.  One
of them is~\cite{IKMW07} with its focus on the notions of
compatibility and bilinguality that result when one adopts both
available products at an extra cost.  Another one is~\cite{BFO10},
where the authors investigate whether the algorithmic approach
of~\cite{KKT03} can be extended to the case of two products.

In \cite{AM11} we introduced a new threshold model of a social network
in which nodes (agents) influenced by their neighbours can adopt one
out of \emph{several} products. This model allowed us to study various
aspects of the spread of a given product through a social network, in
the presence of other products.  We analysed from the complexity point
of view the problems of determining whether adoption of a given
product by the whole network is possible (respectively, necessary), and
when a unique outcome of the adoption process is guaranteed.  We also
clarified for social networks without unique outcomes the complexity
of determining whether a given node has to adopt some (respectively, a
given) product in some (respectively, all) final network(s), and the
complexity of computing the minimum and the maximum possible spread of
a given product.

\subsection{Motivation}
\label{subsec:motiv}

We are interested in understanding and predicting the behaviour of the
consumers (agents) who form a social network and are confronted with
several alternatives (products).  To carry out such an analysis we use
the above model of~\cite{AM11} and associate with each such social
network a natural strategic game. In this game the strategies of an
agent are products he can choose. Additionally a `null' strategy is
available that models the decision of not choosing any product.  The
idea is that after each agent chose a product, or decided not to
choose any, the agents assess the optimality of their choices
comparing them to the choices made by their neighbours.  This leads to
a natural study of (pure) Nash equilibria, in particular of those in
which some, respectively all, constituent strategies are non-null.

Social network games are related to graphical games of
\cite{KLS01}, in which the payoff function of each player depends only
on a (usually small) number of other players.
In this work the focus was mainly on finding mixed (approximate) 
Nash equilibria.  However, in graphical
games the underlying structures are undirected graphs.  Also,
social network games exhibit the following \bfe{join the crowd} property: 

\begin{quote}
  the payoff of each player weakly increases when more players choose his strategy.
\end{quote}
(We define this property more precisely in Subsection~\ref{subsec:sng}.)

Since these games are related to social networks, some natural special
cases are of interest: when the underlying graph is a DAG or has no
source nodes, with the special case of a simple cycle. Such social
networks correspond respectively to a hierarchical organization or to
a `circle of friends', in which everybody has a friend (a
neighbour). Studying Nash equilibria of these games and various
properties defined in terms of improvement paths allows us to gain
better insights into the consequences of adopting products.

%% As noticed in a number of 
%% empirical studies, an abundance of choices may sometimes lead to wrong decisions.
%% To quote from \cite[page 38]{Gig08}:
%% \begin{quote}
%%   \emph{The freedom-of-choice paradox}. The more options one has, the more
%% possibilities for experiencing conflict arise, and the more difficult
%% it becomes to compare the options. There is a point where more options, products, and 
%% choices hurt both seller and consumer.
%% \end{quote}

%% In our model we can analyse a similar phenomenon. Namely, we can
%% explain how the availability of new products to the members of a
%% social network can in some cases permanently destroy market stability.

\subsection{Related work}

There are a number of papers that focus on games associated with
various forms of networks, see, e.g., \cite{TW07} for an overview.  A
more recent example is~\cite{AFPT10} that analyses a strategic game
between players being firms who select nodes in an undirected graph in
order to advertise competing products via `viral marketing'.  However,
in spite of the focus on similar questions concerning the existence
and structure of Nash equilibria and on their reachability, from a
technical point of view, the games studied here seem to be unrelated
to the games studied elsewhere.

Still, it is useful to mention the following phenomenon.  When the
underlying graph of a social network has no source nodes, the game
always has a trivial Nash equilibrium in which no agent chooses a
product. A similar phenomenon has been recently observed
in~\cite{BK11} in the case of their network formation games, where
such equilibria are called degenerate.  Further, note that the `join
the crowd' property is exactly the opposite of the defining property
of the congestion games with player-specific payoff functions
introduced in~\cite{Mil96}. In these game the payoff of each player
weakly decreases when more players choose his strategy.  Because in
our case (in contrast to~\cite{Mil96}) the players can have different
strategy sets, the resulting games are not coordination games.

\subsection{Plan of the paper and summary of the results}

In the next section we recall the model of social networks introduced
in~\cite{AM11} and define strategic games associated with these
networks.  Next, in Section~\ref{sec:general} we show that in general
Nash equilibria do not need to exist even if we limit ourselves to a
special class of networks in which for each node all its neighbours
have the same weight.  We prove that determining an existence of a
Nash equilibrium is NP-complete, also when we limit our attention to
the two special types of equilibria.  Then in
Section~\ref{sec:polymatrix} we show that this NP-completeness result
holds for a more  general class of games, namely polymatrix games.

Motivated by these results we consider in Section~\ref{sec:Ne-special}
strategic games associated with three classes of social networks, the
ones whose underlying graph is a DAG, a simple cycle, or, more
generally, has no source nodes.  For each class we determine the
complexity of deciding whether a Nash equilibrium (possibly of a
special type) exists.  We also show that for these games the price of
anarchy and the price of stability are unbounded.

Next, in Section~\ref{sec:FBRP} we study the finite best response
property (FBRP) of \cite{MS96}. We prove that deciding whether a game
associated with a social network has the FBRP is co-NP-hard.  Then, in
Section~\ref{sec:fbrp-special}, we consider the above three classes of
games. We show that when the underlying graph is a DAG or there are
just two nodes, the game has the FBRP, that there is an efficient
algorithm when the underlying graph is a simple cycle, and that the
problem is co-NP-hard when the underlying graph has no source
nodes. In Sections \ref{sec:FIP} and \ref{sec:fip-special} we obtain
analogous results for the finite improvement property (FIP), though
the complexity when the underlying graph is a simple cycle remains
open.

In Section~\ref{sec:uniform-FIP} we introduce a new property, that we
call \emph{the uniform FIP}, that is of independent interest.  We show
that deciding whether a game associated with a social network has the
uniform FIP is co-NP-hard.  Then, we study in
Section~\ref{sec:uniform-FIP-special} the special cases.  We show that
when the underlying graph is a simple cycle the game has the uniform
FIP. However, when the underlying graph has no source nodes the
problem is co-NP-hard. In Section \ref{sec:WA} we show that the
property of having the uniform FIP is stronger than that of being
weakly acyclic (see \cite{You93} and \cite{Mil96}). Determining
whether a game is weakly acyclic is also co-NP-hard, also when the
underlying graph has no source nodes

Finally, in Section~\ref{sec:conc} we summarize the obtained
complexity results and suggest some further research.

\section{Preliminaries}
\label{sec:prelim}

\subsection{Strategic games}
\label{subsec:games}
Assume a set $\{1, \ldots, n\}$ of players, where $n > 1$.  A
\bfe{strategic game} for $n$ players, written as $(S_1, \ldots, S_n,
p_1, \ldots, p_n)$, consists of a non-empty set $S_i$ of
\bfe{strategies} and a \bfe{payoff function} $p_i : S_1 \times \ldots
\times S_n \: \rightarrow \: \mathbb{R}$,
for each player $i$.

Fix a strategic game
$
G := (S_1, \ldots, S_n, p_1, \ldots, p_n).
$
We denote $S_1 \times \cdots \times S_n$ by $S$, 
call each element $s \in S$
% \bfe{joint strategy},or 
a \bfe{joint strategy},
denote the $i$th element of $s$ by $s_i$, and abbreviate the sequence
$(s_{j})_{j \neq i}$ to $s_{-i}$. Occasionally we write $(s_i,
s_{-i})$ instead of $s$.  
% Finally, we abbreviate $\times_{j \neq i}
% S_j$ to $S_{-i}$.

We call a strategy $s_i$ of player $i$ a \bfe{best response} to a
joint strategy $s_{-i}$ of his opponents if $ \fa s'_i \in S_i
\  p_i(s_i, s_{-i}) \geq p_i(s'_i, s_{-i})$. We call a joint strategy
$s$ a \bfe{Nash equilibrium} if each $s_i$ is a best response to
$s_{-i}$, that is, if
\[
\fa i \in \{1, \ldots, n\} \ \fa s'_i \in S_i \ p_i(s_i, s_{-i}) \geq p_i(s'_i, s_{-i}).
\]
Further, we call a strategy $s_i'$ of player $i$ a \bfe{better
  response} given a joint strategy $s$ if $p_i(s'_i, s_{-i}) >
p_i(s_i, s_{-i})$.

Given a joint strategy $s$ we call the sum $\mathit{SW}(s)=\sum_{j =
  1}^{n} p_j(s)$ the \bfe{social welfare} of $s$.  When the social
welfare of $s$ is maximal we call $s$ a \bfe{social optimum}.  Recall
that, given a finite game that has a Nash equilibrium, its \bfe{price
  of anarchy} (respectively, \bfe{price of stability}) is the ratio
$\frac{\mathit{SW}(s)}{\mathit{SW}(s')}$ where $s$ is a social optimum
and $s'$ is a Nash equilibrium with the lowest (respectively, highest)
social welfare. In the case of division by zero, we interpret the
outcome as $\infty$.

Following the terminology of \cite{MS96}, a \bfe{path} in $S$ is a
sequence $(s^1, s^2, \LL)$ of joint strategies such that for every $k
> 1$ there is a player $i$ such that $s^k = (s'_i, s^{k-1}_{-i})$ for
some $s'_i \neq s^{k-1}_{i}$.  A path is called an \bfe{improvement
  path} if it is maximal and for all $k > 1$, $p_i(s^k) >
p_i(s^{k-1})$ where $i$ is the player who deviated from $s^{k-1}$. If
an improvement path satisfies the additional property that $s^k_i$ is
a best response to $s^{k-1}_{-i}$ for all $k >1$ then it is called a
\bfe{best response improvement path}.

%% A path $\xi$ is called an
%% \bfe{improvement path} if it is maximal and for all $k$ smaller than
%% the length of $\xi$, $p_i(s^k) > p_i(s^{k-1})$, where $i$ is the
%% player who deviated from $s^{k-1}$.  
The last condition simply means that each deviating player selects a
better (best) response.  
A game has the \bfe{finite improvement property} (\bfe{FIP}),
(respectively, the \bfe{finite best response property} (\bfe{FBRP}))
if every improvement path (respectively, every best response
improvement path) is finite. Obviously, if a game has the FIP or the
FBRP, then it has a Nash equilibrium --- it is the last element of
each path. Finally, a game is called \bfe{weakly acyclic} (see
\cite{You93,Mil96}) if for every joint strategy there exists a finite
improvement path that starts at it.

\subsection{Social networks}

We are interested in specific strategic games defined over social networks.
In what follows we focus on a model of the social networks recently introduced
in \cite{AM11}.

Let $V=\{1,\ldots,n\}$ be a finite set of \bfe{agents} and
$\sgraph=(V,E,\weight)$ a weighted directed graph with $\weight_{ij}
\in [0,1]$ being the weight of the edge $(i,j)$. We assume that
$\sgraph$ does not have self loops, i.e., for all $i \in
\{1,\ldots,n\}$, $(i,i) \not\in E$. We often use the notation $i \to
j$ to denote $(i,j) \in E$ and write $i \to^*j$ if there is a path
from $i$ to $j$ in the graph $\sgraph$. Given a node $i$ of $G$ we
denote by $\neighbour(i)$ the set of nodes from which there is an
incoming edge to $i$.  We call each $j \in \neighbour(i)$ a
\oldbfe{neighbour} of $i$ in $G$.  We assume that for each node $i$
such that $\neighbour(i) \neq \ES$, $\sum_{j \in \neighbour(i)} w_{ji}
\leq 1$.  An agent $i \in V$ is said to be a \bfe{source node} in
$\sgraph$ if $\neighbour(i)=\emptyset$.

Let $\products$ be a finite set of alternatives or \bfe{products}.  By
a \bfe{social network} (from now on, just \bfe{network}) we mean a
tuple $\snet=(\sgraph,\products,\prodset,\theta)$, where $\prodset$
assigns to each agent $i$ a non-empty set of products $\prodset(i)$
from which it can make a choice. $\theta$ is a \bfe{threshold
  function} that for each $i \in V$ and $t \in \prodset(i)$ yields a
value $\theta(i,t) \in (0,1]$.
% The threshold $\theta(i,t)$ should be viewed as agent $i$'s
% resistance level to adopt a product $t$. 

Given a network $\snet$ we denote by $\srcnodes(\snet)$ the set of
source nodes in the underlying graph $\sgraph$.  One of the classes of
the networks we shall study are the ones with $\srcnodes(\snet) =
\ES$. We call a network \bfe{equitable} if the weights are defined as
$w_{ji}=\frac{1}{|\neighbour(i)|}$ for all nodes $i$ and $j \in
\neighbour(i)$.

\subsection{Social network games}
\label{subsec:sng}

Fix a network $\snet=(\sgraph,\products,\prodset,\theta)$.
%% We assume that for all $i \in \srcnodes(\snet)$,
%% $|\prodset(i)|=1$. In other words, if $i$ is a source node then
%% $\prodset(i)$ consists of a single product in $\products$.
Each agent can adopt a product from his product set or
choose not to adopt any product. We denote the latter choice by
$t_0$. 

With each network $\snet$ we associate a strategic game
$\mathcal{G}(\snet)$. The idea is that the nodes
simultaneously choose a product or abstain from choosing any.
Subsequently each node assesses his choice by comparing it with the
choices made by his neighbours.  Formally, we define the game as
follows:

\begin{itemize}
\item the players are the agents,

\item the set of strategies for player $i$ is
$S_i :=\prodset(i) \cup \{t_0\}$,

\item for $i \in V$, $t \in \prodset(i)$ and a joint strategy
  $\strprofile$, let $ \inflset_i^t(\strprofile) :=\{j \in
  \neighbour(i) \mid s_j=t\}, $ i.e., $\inflset_i^t(\strprofile)$ is
  the set of neighbours of $i$ who adopted the product $t$ in $s$.

The payoff function is defined as follows, where $\constutil$ is some
positive constant given in advance:

\begin{itemize}
\item for $i \in \srcnodes(\snet)$,

$\payoff_i(\strprofile) :=\left \{\begin{array}{ll}
                         0          & \mbox{if~~} \strprofile_i = t_0\\ 
                         \constutil & \mbox{if~~} \strprofile_i \in \prodset(i)\\
  \end{array}
  \right. $\\

\item for $i \not\in \srcnodes(\snet)$,

$\payoff_i(s) :=\left\{\begin{array}{ll}
               0 &\mbox{if~~} \strprofile_i = t_0\\
               (\sum\limits_{j \in \inflset_i^t(\strprofile)} w_{ji})-\theta(i,t) & \mbox{if~~} \strprofile_i=t, \mbox{ for some } t \in \prodset(i)\\
\end{array}     
                           \right. $\\

\end{itemize}

\end{itemize}

Let us explain the underlying motivations behind the above definition.
In the first item we assume that the payoff function for the source
nodes is constant only for simplicity.  In the last section of the
paper we explain that the obtained results hold equally well in the
case when the source nodes have arbitrary positive utility for each
product.

The second item in the payoff definition is motivated by the following
considerations.  When agent $i$ is not a source node, his
`satisfaction' from a joint strategy depends positively on the
accumulated weight (read: `influence') of his neighbours who made the
same choice as he did, and negatively from his threshold level (read:
`resistance') to adopt this product.  The assumption that $\theta(i,t)
> 0$ reflects the view that there is always some resistance to adopt a
product. So when this resistance is high, it can happen that the
payoff is negative. Of course, in such a situation not adopting any
product, represented by the strategy $t_0$, is a better alternative.

The presence of this possibility allows each agent to refrain from
choosing a product.  This refers to natural situations, such as
deciding not to purchase a smartphone or not going on vacation. In the
last section we refer to an initiated research on social network games
in which the strategy $t_0$ is not present.  Such games capture
situations in which the agents have to take some decision, for
instance selecting a secondary school for their children.

By definition the payoff of each player depends only on the strategies
chosen by his neighbours, so social network games are related to
graphical games of \cite{KLS01}. However, the underlying dependence
structure of a social network game is a directed graph and the
presence of the special strategy $t_0$ available to each player makes
these games more specific. Also, as already mentioned in
Subsection~\ref{subsec:motiv}, these games satisfy the \bfe{join the
  crowd} property that we define as follows:

\begin{quote}
Each payoff function $p_i$ depends only on the strategy chosen by player $i$ and
the set of players who also chose his strategy. Moreover, the dependence
on this set is monotonic.
\end{quote}

In what follows for $t \in \products \cup \{t_0\}$ we use the
notation $\obar{t}$ to denote the joint strategy $\strprofile$ where
$\strprofile_j=t$ for all $j \in V$.  This notation is legal only if for
all agents $i$ it holds that $t \in P(i)$.

The presence of the strategy $t_0$ motivates the introduction and
study of special types of Nash equilibria. We say that a Nash
equilibrium $s$ is

\begin{itemize}
\item \bfe{determined} if for all $i$, $s_i \neq t_0$,

\item \bfe{non-trivial} if for some $i$, $s_i \neq t_0$, 

\item \bfe{trivial} if for all $i$, $s_i = t_0$, i.e., $s = \obar{t_0}$.
\end{itemize}

\section{Nash equilibria: general case}
\label{sec:general}

The first natural question that we address is that of the existence of
Nash equilibria in social network games.  We establish the following
result.

\begin{theorem} \label{thm:np}
Deciding whether for a network $\snet$ the game
$\mathcal{G}(\snet)$ has a (respectively, non-trivial) Nash
equilibrium is NP-complete.
\end{theorem}

To prove it we first construct an example of a social network game
with no Nash equilibrium and then use it to determine the complexity
of the existence of Nash equilibria.

\begin{example} \label{exa:nonash}
\rm
Consider the network given in Figure~\ref{fig:noNe1}, where the product
set of each agent is marked next to the node denoting it and
the weights are labels on the edges. The source nodes are represented
by the unique product in the product set. 
\begin{figure}[ht]
\centering
$
\def\objectstyle{\scriptstyle}
\def\labelstyle{\scriptstyle}
\xymatrix@R=20pt @C=30pt{
& & \{t_1\} \ar[d]_{w_1}\\
& &1 \ar[rd]^{w_2} \ar@{}[rd]^<{\{t_1,t_2\}}\\
\{t_2\} \ar[r]_{w_1} &3 \ar[ur]^{w_2} \ar@{}[ur]^<{\{t_2,t_3\}}& &2 \ar[ll]^{w_2} \ar@{}[lu]_<{\{t_1,t_3\}} &\{t_3\} \ar[l]^{w_1}\\
%%\{t_2\} \ar[ru]_{w_1} & & & &\{t_3\} \ar[lu]^{w_1}\\
}$

\caption{\label{fig:noNe1}A network with no Nash equilibrium}
\end{figure}
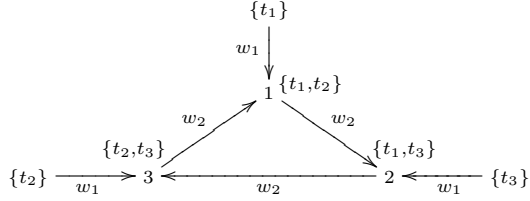

So the weights on the edges from the nodes $\{t_1\}, \{t_2\}, \{t_3\}$
are marked by $w_1$ and the weights on the edges forming the triangle
are marked by $w_2$. We assume that each threshold is a constant
$\theta$, where $ \theta < w_1 < w_2.  $ So it is more profitable to a
player residing on the triangle to adopt the product adopted by his
neighbour residing on a triangle than by the other neighbour who is a
source node. For convenience we represent each joint strategy as a
triple of strategies of players 1, 2 and 3.

It is easy to check that in the game associated with this network no
joint strategy is a Nash equilibrium.  Indeed, each agent residing on
the triangle can secure a payoff of at least $w_1 - \theta>0$, so it
suffices to analyze the joint strategies in which $t_0$ is not
used. There are in total eight such joint strategies. Here is their
listing, where in each joint strategy we underline the strategy that
is not a best response to the choice of other players:
$(\underline{t_1}, t_1, t_2)$, $(t_1, t_1, \underline{t_3})$, $(t_1,
t_3, \underline{t_2})$, $(t_1, \underline{t_3}, t_3)$, $(t_2,
\underline{t_1}, t_2)$, $(t_2, \underline{t_1}, t_3)$, $(t_2, t_3,
\underline{t_2})$, $(\underline{t_2}, t_3, t_3)$.  \HB
\end{example}

\NI
\emph{Proof of Theorem~\ref{thm:np}.}

As in \cite{AM11} we use a reduction from the NP-complete PARTITION
problem, which is: given $n$ positive rational numbers
$(a_1,\LL,a_n)$, is there a set $S$ such that $\sum_{i\in S} a_i =
\sum_{i\not\in S} a_i$?  Consider an instance $I$ of PARTITION.
Without loss of generality, suppose we have normalised the numbers so
that $\sum_{i=1}^n a_i = 1$. Then the problem instance sounds: Is
there a set $S$ such that $\sum_{i\in S} a_i = \sum_{i\not\in S} a_i =
\frac{1}{2}$?
  
To construct the appropriate network we employ the networks given in
Figure~\ref{fig:noNe1} and in Figure~\ref{fig:partition}, where for
each node $i\in\{1,\LL,n\}$ we set $w_{i a} = w_{i b} = a_i$, and
assume that the threshold of the nodes $a$ and $b$ is constant and
equals $\frac12$.

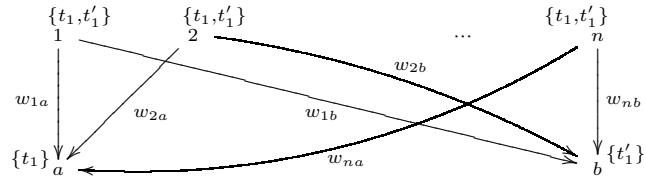
\begin{figure}[ht]
\centering
$
\def\objectstyle{\scriptstyle}
\def\labelstyle{\scriptstyle}
\xymatrix@W=10pt @R=40pt @C=35pt{
1 \ar@{}[r]^<{\{t_1,t_1'\}} \ar[d]_{w_{1a}} \ar[rrrrd]_{w_{1b}}& 2 \ar@{}[r]^<{\{t_1,t_1'\}} \ar[ld]^{w_{2a}} \ar@/^0.7pc/[rrrd]^{w_{2b}}&  &\cdots &n \ar@{}[l]_<{\{t_1,t_1'\}} \ar@/^1.5pc/[lllld]^{w_{na}} \ar[d]^{w_{nb}}\\
a \ar@{}[u]^<{\{t_1\}}  & & & &b \ar@{}[u]_<{\{t_1'\}} \\
}$
\caption{\label{fig:partition}A network related to the PARTITION problem}
\end{figure}

We use two copies of the network given in
Figure~\ref{fig:noNe1}, one unchanged and the other in which the
product $t_1$ is replaced by $t'_1$, and
construct the desired network
$\snet$ by identifying the node $a$ of the network from
Figure~\ref{fig:partition} with the node marked by $\{t_1\}$ in the
network from Figure~\ref{fig:noNe1}, and the node $b$ with the node marked
by $\{t'_1\}$ in the modified version of the network from
Figure~\ref{fig:noNe1}.

% Next, take the network given in
% Figure~\ref{fig:noNe1} and construct the desired network $\snet$
% by identifying the node $a$ of the network from Figure~\ref{fig:partition}
% with the node marked by $\{t_1\}$ in the network in Figure~\ref{fig:noNe1}.

Suppose now that a solution to the considered instance of the
PARTITION problem exists, that is, for some set $S \sse \{1, \LL, n\}$
we have $\sum_{i\in S} a_i = \sum_{i\not\in S} a_i = \frac{1}{2}$.
Consider the game $\mathcal{G}(\snet)$. 
Take the joint strategy formed by the following strategies:

\begin{itemize}

\item $t_1$ assigned to each node $i \in S$ in
the network from Figure~\ref{fig:partition},

\item $t'_1$ assigned to each node $i \in \{1, \LL, n\} \setminus S$
in the network from Figure~\ref{fig:partition},

\item $t_0$ assigned to the nodes $a$, $b$ and the nodes 1 in both versions of
the network from Figure~\ref{fig:noNe1},

\item $t_3$ assigned to the nodes 2, 3 in both versions
of the networks from Figure~\ref{fig:noNe1}
and the two nodes marked by $\{t_3\}$,

\item $t_2$ assigned to the nodes marked by $\{t_2\}$.

\end{itemize}

We claim that this joint strategy is a non-trivial Nash
equilibrium. Consider first the player (i.e, node) $a$. The
accumulated weight of its neighbours who chose strategy $t_1$ is
$\frac12$, so its payoff after switching to the strategy $t_1$ is 0.
Therefore $t_0$ is indeed a best response for player $a$.  For the
same reason, $t_0$ is also a best response for player $b$.  The
analysis for the other nodes is straightforward and left to the
reader.

Conversely, suppose that a joint strategy $s$ is a Nash equilibrium in
the game $\mathcal{G}(\snet)$. Then it is also a non-trivial Nash
equilibrium since for each source node $i \in \{1,\ldots,n\}$,
$\strprofile_i \neq t_0$. We claim that the strategy selected by the
node $a$ in $s$ is $t_0$. Otherwise, this strategy equals $t_1$ and
the strategies selected by the nodes of the network of
Figure~\ref{fig:noNe1} form a Nash equilibrium in the game associated
with this network. This yields a contradiction with our previous
analysis of this network.

So $t_0$ is a best response of the node $a$ to the strategies of the
other players chosen in $s$.  This means that $ \sum_{i \in \{1, \LL,
  n\} \mid s_i = t_1} w_{i a} \leq \frac12$. By the same reasoning
$t_0$ is a best response of the node $b$ to the strategies of the
other players chosen in $s$.  This means that $\sum_{i \in \{1, \LL,
  n\} \mid s_i = t'_1} w_{i b} \leq \frac12$.

But $\sum_{i=1}^n a_i = 1$ and for $i\in\{1,\LL,n\}$, $w_{i a} = w_{i
  b} = a_i$, and $s_i \in \{t_1, t'_1\}$ since $i$ is a source node.
So both above inequalities are in fact equalities.  Consequently for
$S := \{i \in \{1, \LL, n\} \mid s_i = t_1\}$ we have $\sum_{i\in S}
a_i = \sum_{i\not\in S} a_i$.  In other words, there exists a solution
to the considered instance of the PARTITION problem.

Finally, given an instance of the PARTITION problem, 
the above network $\snet$ can be constructed from it
in polynomial time. This proves the NP-hardness
of the considered problem.

To prove that the problem lies in NP it suffices to notice that given
a network $\snet=(\sgraph,\products,\prodset,\theta)$ with $n$
nodes checking whether a joint strategy is a non-trivial Nash
equilibrium can be done by means of $n \cdot |\products|$ checks, so
in polynomial time.  \qed

%% \HB

\VV
  
%% The proof remains the same when we consider the
%% complexity of the existence of a non-trivial Nash equilibrium.
Focussing on determined Nash equilibria does not change the matters.

\begin{theorem} \label{thm:nd}
Deciding whether for a network $\snet$ the game
$\mathcal{G}(\snet)$ has a determined Nash equilibrium is NP-complete.
\end{theorem}
\begin{proof}
We use an instance of the PARTITION problem in the form of $n$
positive rational numbers $(a_1,\LL,a_n)$, normalised as in the
previous proof, and the network given in Figure~\ref{fig:partition}.

Suppose now that a solution to the considered instance of the
partition problem exists, that is for some set $S \sse \{1, \LL, n\}$
we have $\sum_{i\in S} a_i = \sum_{i\not\in S} a_i = \frac{1}{2}$.
Take the joint strategy $s$ formed by the following strategies:

\begin{itemize}

\item $t_1$ assigned to each node $i \in S$ and the node $a$,

\item $t'_1$ assigned to each node $i \in \{1, \LL, n\} \setminus S$
and the node $b$.

\end{itemize}

Then $p_a(t_1, s_{-i}) = \frac12 - \theta(a, t_1) = 0$, so  $p_a(t_1, s_{-i}) \geq  p_a(t_0, s_{-i})$.
Analogously $p_b(t'_1, s_{-i}) \geq  p_b(t_0, s_{-i})$. So $s$ is a determined Nash equlibrium
in the strategic game associated with the above network.

Consider now a determined Nash equilibrium $s$ in this game.  Then
$s_a = t_1$ and $p_a(t_1, s_{-i}) \geq p_a(t_0, s_{-i}) = 0$.  So for
$S := \{i \in \{1, \LL, n\} \mid s_i = t_1\}$ we have $\sum_{i\in S}
a_i \geq \frac12$.  Analogously $\sum_{i \not\in S} a_i \geq
\frac12$. Since $\sum_{i=1}^n a_i = 1$, in both cases we have in fact
equalities and hence there exists a solution to the considered
instance of the PARTITION problem.
\end{proof}

Recall that the social network in Example \ref{exa:nonash} used three
products.  The following result shows that to construct a social
network $\snet$ such that $\mathcal{G}(\snet)$ has no Nash equilibrium
in fact at least three products are required.

\begin{theorem}
For a network $\snet$, if there exists a non-empty set $X \subseteq
\products$ such that $|X| \leq 2$ and for all $i \in
\srcnodes(\snet)$, $\prodset(i) \cap X \neq \emptyset$ then
$\mathcal{G}(\snet)$ has a Nash equilibrium and it can be computed in
polynomial time.
\end{theorem}

In particular $\mathcal{G}(\snet)$ has a Nash equilibrium when all
nodes $i$ have the same set of two products.

\begin{proof}
Given an initial joint strategy we call a maximal sequence of best
response deviations to a given strategy $t$ (in an arbitrary order) a
\bfe{$t$-phase}. Let $\snet=(\sgraph,\products,\prodset,\theta)$ where
$\sgraph=(V,E,\weight)$.

First, suppose that $|X|=1$, say $X=\{t_1\}$. Let $\strprofile$ be the
resulting joint strategy after performing a $t_1$-phase starting in
the joint strategy $\obar{t_0}$. We show that $\strprofile$ is a Nash
equilibrium. First note that $\strprofile_j = t_1$ for every $j \in
\srcnodes(\snet)$. Further, in the $t_1$-phase, if a joint strategy
$s^2$ is obtained from $s^1$ by having some nodes switch to product
$t_1$ and $t_1$ is a best response for a node $i$ to $s^1_{-i}$, then
$t_1$ remains a best response for $i$ to $s^2_{-i}$. Indeed, by the
join the crowd property $p_{i}(t_1, s^2_{-i}) \geq p_{i}(t_1,
s^1_{-i})$ and $p_{i}(t_1, s^1_{-i}) \geq p_{i}(t_0, s^1_{-i})$, so
$p_{i}(t_1, s^2_{-i}) \geq p_{i}(t_0, s^2_{-i})$.  Consequently after
the first $t_1$-phase, in the resulting joint strategy $\strprofile$,
each node that has the strategy $t_1$ plays a best response. If for
some $j$, $\strprofile_j = t_0$ then by the definition of the
$t_1$-phase, $j$ is playing his best response as well. Therefore
$\strprofile$ is a Nash equilibrium.

Now suppose that $|X|=2$, say $X=\{t_1,t_2\}$. Let $V_{t_1}=\{j \in
\srcnodes(\snet) \mid t_1 \in \prodset(j)\}$ and $\obar{V}_{t_1}=\{j
\in \srcnodes(\snet) \mid t_1 \not\in \prodset(j)\}$. 
%% So
%% $\srcnodes(\snet) \setminus \obar{V}_{t_1}$ is the set of source nodes
%% $i$ such that $t_1 \in \prodset(i)$.
Let $\snet^{t_1}=(\sgraph^{t_1},\products,\prodset,\theta)$, where
$\sgraph^{t_1}$ is the induced subgraph of $\sgraph$ on the nodes $V
\setminus \obar{V}_{t_1}$. Let $\strprofile^{t_1}$ be the resulting
joint strategy in $\snet^{t_1}$ after performing a $t_1$-phase
starting in $\obar{t_0}$. By the previous argument,
$\strprofile^{t_1}$ is a Nash equilibrium in
$\mathcal{G}(\snet^{t_1})$. Now consider the joint strategy
$\strprofile$ in $\mathcal{G}(\snet)$ defined as follows:
\[
\strprofile_i = 
\begin{cases}
t_0 &\mathrm{if}\ i \in \obar{V}_{t_1}\\
\strprofile^{t_1}_i & \mathrm{otherwise}
\end{cases}
\]

Starting at $\strprofile$, we repeatedly perform a $t_2$-phase
followed by a $t_0$-phase. We claim that this process terminates in a
Nash equilibrium in $\mathcal{G}(\snet)$.

First note that if a joint strategy $s^2$ is obtained from $s^1$ by
having some nodes switch to product $t_2$ and $t_2$ is a best response
for a node $i$ to $s^1_{-i}$, then $t_2$ remains a best response for
$i$ to $s^2_{-i}$. The argument is analogous to the one in the previous
case. Therefore after the first $t_2$-phase each node that has the
strategy $t_2$ plays a best response. Call the outcome of the first
$t_2$-phase $s''$.

Now consider a node $i$ that deviated to $t_0$ starting at $s''$ by
means of a best response. By the observation just made, node $i$
deviated from product $t_1$. So, again by the join the crowd property,
this deviation does not affect the property that the nodes that
selected $t_2$ in $s''$ play a best response.  Iterating this
reasoning we conclude that after the first $t_0$-phase each node that
has the strategy $t_2$ continues to play a best response.

By the same reasoning subsequent $t_2$ and $t_0$-phases have the same
effect on the set of nodes that have the strategy $t_2$, namely that
each of these nodes continues to play a best response.

Moreover, this set continues to weakly increase.  Consequently these
repeated applications of the $t_2$-phase followed by the $t_0$-phase
terminate, say in a joint strategy $s'$.  Now suppose that a node $i$
does not play a best response to $s'_{-i}$. Then clearly $i \not\in
\srcnodes(\snet)$. If $s'_i = t_0$, then by the construction $t_2$ is
not a best response, so $t_1$ is a best response.

Suppose $s_i = t_0$. Consider the joint strategy $\strprofile^{t_1}$
which is a Nash equilibrium in $\mathcal{G}(\snet^{t_1})$. We have
$p_i(t_1, s_{-i}^{t_1}) \leq p_i(t_0, s_{-i}^{t_1})$. Since $i \not\in
\srcnodes(\snet)$, we have $\strprofile_i =
\strprofile_i^{t_1}$. Since for all $j \in \obar{V}_{t_1}$, $t_1
\not\in \prodset(j)$ we have $p_i(t_1, s_{-i}) \leq p_i(t_0, s_{-i})$
as well. By the join the crowd property $p_i(t_1, s'_{-i}) \leq
p_i(t_1, s_{-i})$, so $p_i(t_1, s'_{-i}) \leq p_i(t_0, s'_{-i})$,
which yields a contradiction.  Hence node $i$ deviated to $t_0$ from
some intermediate joint strategy $s^1$ by selecting a best response.
So $p_i(t_1, s^1_{-i}) \leq p_i(t_0, s^1_{-i})$. Moreover, by the join
the crowd property $p_i(t_1, s'_{-i}) \leq p_i(t_1, s^1_{-i})$, so
$p_i(t_1, s'_{-i}) \leq p_i(t_0, s'_{-i})$, which yields a
contradiction, as well.

Further, by the construction $s'_i \neq t_2$, so the only alternative
is that $s'_i = t_1$. But then either $t_0$ or $t_2$ is a best
response, which contradicts the construction of $s'$.  We conclude
that $s'$ is a Nash equilibrium in $\mathcal{G}(\snet)$.

Finally, note that each $t$-phase takes at most $n$ steps to
terminate, where $n$ is the number of nodes in $\sgraph$. Since in
each repeated iteration of a $t_2$-phase followed by a $t_0$-phase,
the set of nodes which chose $t_2$ is weakly increasing, this process
terminates after at most $\mathcal{O}(n)$ iterations. Therefore this
Nash equilibrium can be computed in polynomial time.
\end{proof}

\section{Polymatrix games}
\label{sec:polymatrix}
We now apply the results of the previous section to a class of games
that are more general than social network games.

A polymatrix game, see \cite{Jan68,How72}, is a finite strategic game
in which the influence of a pure strategy selected by any player on
the payoff of any other player is always the same, regardless what
strategies other players select. Formally, it is a game
$(S_1,\ldots,S_n,\payoff_1,\ldots,\payoff_n)$ in which for all pairs of
players $i$ and $j$ there exists a partial payoff function
$a^{ij}$ such that for any joint strategy
$\strprofile=(\strprofile_1,\ldots,\strprofile_n)$, the payoff of
player $i$ is given by $\payoff_i(\strprofile):=\sum_{j \neq i}
a^{ij}(\strprofile_i,\strprofile_j)$.

\begin{theorem}
\label{thm:polymatrix-NP}
Deciding whether a polymatrix game has a Nash equilibrium is NP-complete.
\end{theorem}

\begin{proof}
It is clear that the problem is in NP since one can guess a joint
strategy and check whether it is a Nash equilibrium in polynomial
time. To prove the hardness result we show that social network games
are a subclass of polymatrix games. The claim then follows from
Theorem~\ref{thm:np}.

For a social network $\snet=(\sgraph,\products,\prodset,\theta)$ with
$n$ agents, let the associated game be
$\mathcal{G}(\snet)=(S_1,\ldots,S_n,\payoff_1,\ldots,\payoff_n)$. We
translate $\mathcal{G}(\snet)$ into a payoff equivalent polymatrix
game $\mathcal{G}'=(S_1,\ldots,S_n,\payoff_1',\ldots,\payoff_n')$ as
follows. For a joint strategy $\strprofile$ and players $i$ and $j$,
the partial payoff functions are given by

\begin{itemize}
\item for $i \in \srcnodes(\snet)$,

  $a^{ij}(s_i,s_j) := \begin{cases}
                    0 & \text{ if } \strprofile_i=t_0,\\
                    \frac{\constutil}{n-1} & \text{ if } \strprofile_i \in \prodset(i),
                    \end{cases}$

\item for $i \not\in \srcnodes(\snet)$,

$a^{ij}(s_i,s_j) :=\begin{cases}
    0 & \text{ if } \strprofile_i=t_0,\\
    w_{ji}-\frac{\theta(i,\strprofile_i)}{n-1} & \text{ if } j \in \neighbour(i), s_i \neq t_0 \mbox{ and } \strprofile_i=\strprofile_j,\\
    -\frac{\theta(i,\strprofile_i)}{n-1} & \mbox{ otherwise}.
    \end{cases}$

\end{itemize}

The payoff function for player $i$ is then defined as
$\payoff'_i(\strprofile)=\sum_{j \neq i}
a^{ij}(\strprofile_i,\strprofile_j)$. To complete the proof, we show
that for all $i \in \{1,\ldots,n\}$, for all joint strategies
$\strprofile$, we have
$\payoff_i'(\strprofile)=\payoff_i(\strprofile)$. Recall that for a
player $i$ and a joint strategy $\strprofile$,
$\inflset_i^{\strprofile_i}(\strprofile)=\{j \in \neighbour(i) \mid
\strprofile_i=\strprofile_j\}$. We have the following cases:

\begin{itemize}
\item $\strprofile_i=t_0$. Then
  $\payoff_i'(\strprofile)=0=\payoff_i(\strprofile)$.
\item $\strprofile_i \neq t_0$ and $i \in \srcnodes(\snet)$.

Then  $\payoff_i'(\strprofile)=\sum_{j \neq i} \frac{\constutil}{n-1}=
  (n-1) \cdot \frac{\constutil}{n-1}=\payoff_i(\strprofile)$.

\item $\strprofile_i \neq t_0$ and $i \not\in \srcnodes(\snet)$.

$\begin{array}{lll}
\text{Then }\payoff_i'(\strprofile) & = &\sum_{j \in
    \inflset_i^{\strprofile_i}(\strprofile)}
  (w_{ji}-\frac{\theta(i,s_i)}{n-1}) + \sum_{j\neq i, j \not\in
    \inflset_i^{\strprofile_i}(\strprofile)}
 -\frac{\theta(i,s_i)}{n-1}\\
  & = & \sum_{j \in \inflset_i^{\strprofile_i}(\strprofile)} w_{ji} - (n-1) \cdot \frac{\theta(i,s_i)}{n-1}\\
  & = & \payoff_i(\strprofile).
\end{array}$
\end{itemize}
\end{proof}

\section{Nash equilibria: special cases}
\label{sec:Ne-special}

In view of the fact that in general Nash equilibria may not exist we
now consider networks with special properties of the underlying
directed graph. We focus on three natural classes.

\subsection{Directed acyclic graphs}
\label{subsec:dag}

We consider first networks whose underlying graph is a directed
acyclic graph (DAG).  Intuitively, such networks correspond to
hierarchical organizations. This restriction leads to a different
outcome in the analysis of Nash equilibria.

Given a DAG $G := (V,E)$, we use 
a fixed level by level
enumeration $\dagrank()$ of its nodes so that
for all $i,j \in V$
\begin{equation}
 \label{equ:rank}
\mbox{if $\dagrank(i) < \dagrank(j)$, then there is
no path in $\sgraph$ from $j$ to $i$.}
\end{equation}

\begin{theorem}
Consider a network $\snet$ whose underlying graph is a DAG.
\begin{enumerate}[(i)]
\item $\mathcal{G}(\snet)$ always has a non-trivial Nash equilibrium.
\item Deciding whether $\mathcal{G}(\snet)$ has a determined Nash
  equilibrium is NP-complete.
\end{enumerate}
\end{theorem}

\begin{proof}
~\\ \NI($i$) We proceed by assigning to each node a strategy following
  the order determined by (\ref{equ:rank}). Given a node we assign to
  it a best response to the sequence of strategies already assigned to
  all his neighbours.  (By definition the strategies of other players
  have no influence on the choice of a best response.) This yields a
  non-trivial Nash equilibrium. This result is also an immediate
  consequence of Theorem~\ref{thm:fip-DAG} proved in
  Section~\ref{sec:FIP}.

\NI ($ii$) The claim follows from Theorem~\ref{thm:nd} since the
underlying graph of the network from Figure~\ref{fig:partition} used
in its proof is a DAG.
\end{proof}

Note that when the underlying graph is a DAG all Nash equilibria are
non-trivial. Further, in the procedure described in ($i$), in general
more than one best response can exist.  In that case multiple Nash
equilibria exist.

The above procedure uses the set of best responses
$BR_i(s_{\neighbour(i)})$ of player $i$ to the joint strategy
$s_{\neighbour(i)}$ of his neighbours in $\mathcal{G}(\snet)$.  This
set is defined directly in terms of $s_{\neighbour(i)}$ and $\snet$ as
follows, where $\snet =(\sgraph,\products,\prodset,\theta)$.

Let
\[
\begin{array}{l}
Z_i^{>0}(s_{\neighbour(i)})  :=\{t \in \prodset(i)
\mid \sum\limits_{k \in \neighbour(i), s_k = t} w_{ki} -
\theta(i,t)>0\} \text{, } \\
Z_i^{=0}(s_{\neighbour(i)}) := \{t_0\} \cup \{t \in \prodset(i) \mid
\sum\limits_{k \in \neighbour(i), s_k = t} w_{ki} - \theta(i,t)=0\}.
\end{array}
\]
Then
\[
BR_i(s_{\neighbour(i)}) := 
\begin{cases}
\argmax\limits_{t \in Z_i^{>0}(s_{\neighbour(i)})} \big{(}\sum\limits_{k \in \neighbour(i), s_k = t}
  w_{ki} - \theta(i,t)\big{)} & \textrm{if } Z_i^{>0}(s_{\neighbour(i)}) \neq \emptyset \\
Z_i^{=0}(s_{\neighbour(i)})                                & \textrm{otherwise}
\end{cases}
\]

\smallskip

Finally, we consider the price of anarchy and the price
of stability for the considered class of games. The following simple
result holds.

\begin{theorem}
\label{thm:poa1}
  The price of anarchy and the price of stability for the games associated
  with the networks whose underlying graph is a DAG is
  unbounded.
\end{theorem}

\begin{proof}
Consider the network depicted in Figure~\ref{fig:poa1}.

\begin{figure}[htb]
\centering
$ \def\objectstyle{\scriptstyle} \def\labelstyle{\scriptstyle}
  \xymatrix@W=10pt @C=40pt @R=5pt{
& & \\
i \ar[r]^{w_{ij}} \ar@{}[r]^<{\{t_1\}}& j \ar[r]^{w_{jk}} \ar@{}[r]^<{\{t_2\}}& k \ar@{}[u]_<{\{t_2\}}\\
}$
\caption{\label{fig:poa1}A network with a high price of anarchy and stability}
\end{figure}
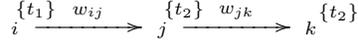

Choose an arbitrary value $r > 0$.  Suppose first that there exists weights and
thresholds such that $\weight_{jk} - (\theta(j,t_2) + \theta(k,t_2)) >
r \constutil$.  (Recall that the payoff of player $i$ is
$\constutil$.)

The game associated with this network has a unique Nash equilibrium,
namely the joint strategy $(t_1, t_0, t_0)$ assigned to the sequence
$(i,j,k)$ of nodes.  Its social welfare is $\constutil$. In contrast, the
social optimum is achieved by the joint strategy $(t_1, t_2, t_2)$ and equals
\[
\constutil + \weight_{jk} - (\theta(j,t_2) + \theta(k,t_2)) > r \constutil.
\]
So for every value $r > 0$ there is a network whose
game has price of anarchy and price of stability higher than $r$.

Suppose now that the inequality $\weight_{jk} - (\theta(j,t_2) +
\theta(k,t_2)) > r \constutil$ does not hold for any choice of weights
and thresholds. (This is for instance the case when $r \constutil \geq
1$, which can be the case as $\constutil$ and $r$ are arbitrary.)  In
that case, we modify the above social network as follows. First, we
replace the node $k$ by $\ceiling{r \constutil + 1}$ nodes, all direct
descendants of node $j$ and each with the product set $\{t_2\}$.  Then
we choose the weights and the thresholds in such a way that the sum of
all these weights minus the sum of all the thresholds for the product
$t_2$ exceeds $r \constutil$. In the resulting game, by the same
argument as above, both the price of anarchy and price of stability
are higher than $r$.
\end{proof}

\subsection{Simple cycles}
\label{subsec:simple}
Next, we consider networks whose underlying graph is a simple cycle.
To fix the notation suppose that the underlying graph is $1 \to 2 \to
\LL \to n \to 1$.  We assume that the counting is done in cyclic order
within \C{\LLn} using the increment operation $i \oplus 1$ and the
decrement operation $i \ominus 1$. In particular, $n \oplus 1 = 1$ and
$1 \ominus 1 = n$.  The payoff functions can then be rewritten as
follows:
\[
p_i(s) := \begin{cases}
        0 & \mathrm{if}\  s_{i} = t_0 \\
        w_{i \ominus 1 \: i} - \theta(i, s_i)    & \mathrm{if}\  s_{i} = s_{i \ominus 1} \ \mathrm{and} \ s_{i} \in P(i) \\
        - \theta(i, s_i)    & \mathrm{otherwise}
        \end{cases}
\]

Clearly $\obar{t_0}$ is a trivial Nash equilibrium.
The following observation clarifies when other Nash equilibria exist.

\begin{theorem} \label{thm:cycle}
Consider a network $\snet$ whose underlying graph is a simple cycle.  
Then $s$ is a non-trivial (respectively, determined) Nash equilibrium
of the game $\mathcal{G}(\snet)$ iff $s$ is of the form $\bar{t}$
for some product $t$ and for all $i$, $p_i(s) \geq 0$.
\end{theorem}
\begin{proof}

\NI
$(\Ra)$ Consider a non-trival Nash equilibrium $s$. Suppose that $s_i = t$
for a product $t$. We have $p_i(s) \geq p_i(t_0, s_{-i}) = 0$, 
so $s_{i \ominus 1} = t$. Iterating this reasoning we conclude that 
$s = \bar{t}$.
\II

\NI
$(\La)$ Straightforward.
\end{proof}

This yields a straightforward algorithm that allows us to
check whether there exists a non-trivial, respectively, determined, Nash equilibrium.

\begin{theorem}
\label{thm:cycle-NE}
Consider a network $\snet=(\sgraph,\products,\prodset,\theta)$ whose underlying graph is a simple cycle.
There is a procedure that runs in time $\bigo(|\products| \cdot n)$, where $n$
is the number of nodes in $\sgraph$, that decides whether
$\mathcal{G}(\snet)$ has a non-trivial (respectively, determined) Nash equilibrium.
\end{theorem}

\begin{proof}
Thanks to Theorem~\ref{thm:cycle} we can use the following procedure.
\[\mathsf{VerifyNashCycle(\snet)}\]

$X := P(1)$;

$\mathit{found}:= \textbf{false}$;

{\bf while} $X \neq \emptyset$ {\bf and} \textbf{not }$\mathit{found}$ {\bf do}

\quad choose $t \in X$; $X := X \setminus \{t\}$;

\quad $i := 1$;

\quad $\mathit{found}:= (w_{i \ominus 1 \: i} \geq \theta(i, t))$;

\quad {\bf while} $i \neq n$ {\bf and} $\mathit{found}$ {\bf do}

\quad \quad $i := i + 1$;

\quad \quad $\mathit{found}:= (t \in P(i)) \wedge (w_{i \ominus 1 \: i} \geq \theta(i, t))$

\quad {\bf od}

%% \quad $\mathit{found}:= \mathit{continue}$;

\textbf{od}

{\bf return} $\mathit{found}$

\smallskip

This procedure returns {\bf true} if a non-trivial (or equivalently, a
determined) Nash equilibrium exists and {\bf false} otherwise. Its
running time is $\bigo(|\products| \cdot n)$.
\end{proof}

Next, we consider the price of anarchy and the price of stability.
We have the following counterpart of Theorem~\ref{thm:poa1}.

\begin{theorem}
\label{thm:poa2}
  The price of anarchy and the price of stability for the games associated
  with the networks whose underlying graph is a simple cycle is
  unbounded.
\end{theorem}

%\begin{proof}
%See the Appendix.
\begin{proof}
Choose an arbitrary value $r > 0$ and let $\epsilon$ be such that 
$\epsilon < min(\frac{1}{4}, \frac{1}{2 (r + 1)})$. Then both
$1 - 2 \epsilon > 2 \epsilon$ and $\frac{1 - 2 \epsilon}{2 \epsilon} > r$.

Consider the network depicted in Figure~\ref{fig:poa-cycle}.

\begin{figure}[ht]
\centering
$ \def\objectstyle{\scriptstyle} \def\labelstyle{\scriptstyle}
  \xymatrix@W=10pt @C=50pt @R=5pt{
& \\
1 \ar@/^0.7pc/[r]^{} \ar@{}[u]^<{\{t_1,t_2\}}& 2 \ar@/^0.7pc/[l]^{} 
\ar@{}[u]_<{\{t_1,t_2\}} 
}$
\caption{\label{fig:poa-cycle}Another network with a high price of anarchy and stability}
\end{figure}
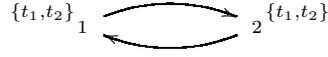

We assume that 
\[
\mbox{
$w_{ 1 2} - \theta(2, t_2) = 1 - \epsilon$,
$w_{ 2 1} - \theta(1, t_2) = - \epsilon$,
$w_{1 2} - \theta(2, t_1) = \epsilon$,
$w_{2 1} - \theta(1, t_1) =  \epsilon$.
}
\]
Then the social optimum is achieved in the joint strategy $(t_2, t_2)$ and
equals $1 - 2 \epsilon$.  There are two Nash equilibria, $(t_1,t_1)$
and the trivial one, with the respective social welfare $2 \epsilon$
and 0.

In the case of the price of anarchy we deal with the division by
zero. We interpret the outcome as $\infty$. The price of stability
equals $\frac{1 - 2 \epsilon}{2 \epsilon}$, so is higher than $r$.
\end{proof}

%\end{proof}

\subsection{Graphs with no source nodes}
\label{subsec:nosource}

Finally, we consider the case when the underlying graph $\sgraph =
(V,E)$ of a network $\snet$ has no source nodes, i.e., for all
$i \in V$, $\neighbour(i) \neq \emptyset$.  Intuitively, such a
network corresponds to a `circle of friends': everybody has a friend
(a neighbour).  For such networks we prove the following
result.

\begin{theorem} \label{thm:nosource1}
Consider a network $\snet=(\sgraph,\products,\prodset,\theta)$ whose underlying graph has no source nodes.
There is a procedure that runs in time $\bigo(|\products| \cdot n^3)$, where $n$
is the number of nodes in $\sgraph$, that decides whether
$\mathcal{G}(\snet)$ has a non-trivial Nash equilibrium.
 \end{theorem}

The proof of Theorem \ref{thm:nosource1} requires some
characterization results that are of independent interest.  The
following concept plays a crucial role.  Here and elsewhere we only
consider subgraphs that are \emph{induced} and identify each such
subgraph with its set of nodes. (Recall that $(V',E')$ is an induced
subgraph of $(V,E)$ if $V' \sse V$ and $E' = E \cap (V' \times V')$.)

We say that a (non-empty) strongly connected subgraph (in short, SCS)
$C_t$ of $\sgraph$ is \bfe{self-sustaining} for a product $t$ if for
all $i \in C_t$,

\begin{itemize}
\item $t \in \prodset(i)$,
\item $\sum\limits_{j \in \neighbour(i)\cap C_t} w_{ji} \geq
  \theta(i,t)$.
\end{itemize}

An easy observation is that if $\snet$ is a network with no source
nodes, then it always has a trivial Nash equilibrium,
$\obar{t_0}$. The following lemma states that for such networks every
non-trivial Nash equilibrium satisfies a structural property which
relates it to the set of self- sustaining SCSs in the underlying
graph. We use the following notation: for a joint strategy
$\strprofile$ and a product $t$, $\agents_t(\strprofile):=\{i \in V
\mid \strprofile_i=t\}$ and $\prodset(\strprofile):=\{t \mid \exists i
\in V \text{ with }s_i=t\}$.

\begin{lemma}
\label{lm:Ne-struct}
Let $\snet=(\sgraph,\products,\prodset,\theta)$ be a network whose
underlying graph has no source nodes. If $\strprofile \neq \obar{t_0}$
is a Nash equilibrium in $\mathcal{G}(\snet)$, then for all products
$t \in \prodset(\strprofile) \setminus \{t_0\}$ and $i \in
\agents_t(s)$ there exists a self-sustaining SCS $C_t \subseteq
\agents_t(\strprofile)$ for $t$ and a $j \in C_t$ such that $j \to^*
i$.
\end{lemma}

\begin{proof}
Suppose $s \neq \obar{t_0}$ is a Nash equilibrium. Take any product $t
\neq t_0$ and an agent $i$ such that $\strprofile_i=t$ (by assumption,
at least one such $t$ and $i$ exists).  Recall that
$\inflset_i^t(\strprofile)$ is the set of neighbours of $i$ who
adopted the product $t$ in $s$.  Consider the set of nodes
$\spred:=\bigcup_{m \in \nat} \spred_m$, where
\begin{itemize}
\item $\spred_0:=\{i\}$,

\item $\spred_{m+1}:=\spred_m \cup \bigcup_{j \in \spred_m} \inflset_j^t(\strprofile)$.

% \item $\spred_{m+1}:=\spred_m \cup \{k \mid \exists j \in \spred_{m}
%   \text{ such that } k \in \neighbour(j) \text{ and }\strprofile_k=t\}$.
\end{itemize}

By construction for all $j \in \spred$, $s_j = t$ and
$\inflset_j^t(\strprofile) \sse \spred$.  Moreover, since
$\strprofile$ is a Nash equilibrium, we also have $\sum\limits_{k \in
  \inflset_j^t(\strprofile)} w_{kj} \geq \theta(j,t)$.

Consider the partial ordering $<$ between the strongly connected
components of the graph induced by $\spred$ defined by: $C < C'$ iff
$j \to k$ for some $j \in C$ and $k \in C'$. Now take some SCS $C_t$
induced by a strongly connected component that is minimal in the $<$
ordering.
%% So $j \to k$ and $k \in C_t$ implies $j \in C_t$.
Then for all $k \in C_t$ we have $\inflset_k^t(\strprofile) \sse C_t$
and hence
$\inflset_k^t(\strprofile) \sse \neighbour(k)\cap C_t$.
This shows that $C_t$ is self-sustaining.

Moreover, by the construction of $\spred$ for all $j \in \spred$, and
a fortiori for all $j \in C_t$, we also have $j \to^* i$. Since the
choice of $t$ and $i$ was arbitrary, the claim follows.
\end{proof}

\begin{figure}[ht]
\centering
$ \def\objectstyle{\scriptstyle} \def\labelstyle{\scriptstyle}
  \xymatrix@W=10pt @C=20pt{
& 1 \ar[rd]^{} \ar@{}[rd]^<{\{t_1,t_2\}}\\
3 \ar[ru]^{} \ar@{}[ru]^<{\{t_1,t_2,t_3,t_4\}}& & 2 \ar[ll]^{} \ar@{}[lu]_<{\{t_1,t_2\}} \ar[ddl]{}\\
\\  %% Separation of cycles
& 4 \ar[dr]{} \ar@{}[dr]^<{\{t_3,t_4\}} \ar[uul]{}\\
6 \ar[ur]{}\ar[uuu]{} \ar@{}[uuu]^<{\{t_3,t_4\}}& & 5 \ar[ll]{} \ar@{}[uuu]_<{\{t_3,t_4\}} \\
}
$
\caption{\label{fig:lem-converse}An equitable network}
\end{figure}
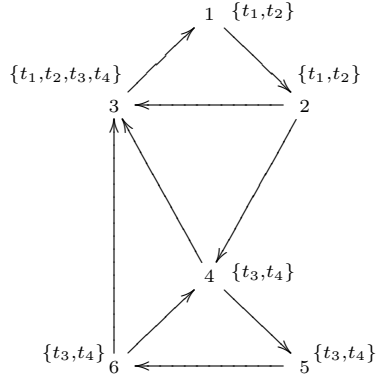 

%% \begin{remark}
%% \label{rem:Ne-struct-conv}
The converse of Lemma~\ref{lm:Ne-struct} does not hold. Indeed,
consider the equitable network given in Figure
\ref{fig:lem-converse}. The product set of each agent is marked next
to the node denoting the agent. Assume that the threshold of each node
is a constant smaller than $\frac{1}{k}$, where $k$ is the number of
incoming edges. So each agent has a non-negative payoff when he adopts
any product adopted by some of his neighbours. Consider the joint
strategy $\strprofile$ in which agents 1, 2 and 3 adopt product $t_1$
and agents 4, 5 and 6 adopt product $t_3$, i.e.\ ,
$\strprofile=(t_1,t_1,t_1,t_3,t_3,t_3)$. It follows from the
definition that $\strprofile$ satisfies the following condition of
Lemma~\ref{lm:Ne-struct}:
\begin{itemize}
\item for all products $t \in \prodset(\strprofile) \setminus \{t_0\}$
  and $k \in \agents_t(s)$ there exists a self-sustaining SCS $C_t
  \subseteq \agents_t(\strprofile)$ for $t$ and $j \in S_t$ such that
  $j \to^* k$.
\end{itemize}

However, $\strprofile$ is not a Nash equilibrium since agent 3 has the
incentive to deviate to product $t_3$. Also, note that the joint strategy
$\strprofile'=(t_0,t_0,t_3,t_3,t_3,t_3)$ is a Nash equilibrium.

%%\end{remark}

Using Lemma~\ref{lm:Ne-struct}, we can now provide a necessary and
sufficient condition for the existence of non-trivial Nash equilibria
for the considered networks.

\begin{lemma}
\label{lm:Ne-nosrc}
Let $\snet=(\sgraph,\products,\prodset,\theta)$ be a network
whose underlying graph has no source nodes. The joint strategy
$\obar{t_0}$ is a unique Nash equilibrium in $\mathcal{G}(\snet)$ iff
there does not exist a product $t$ and a self-sustaining SCS $C_t$ for
$t$ in $\sgraph$.
\end{lemma}

\begin{proof}

\noindent($\Leftarrow$) By Lemma~\ref{lm:Ne-struct}.

%% Suppose there exists a joint strategy
%% $\strprofile \neq \obar{t_0}$ such that $\strprofile$ is a Nash
%% equilibrium. Then by Lemma~\ref{lm:Ne-struct} there exists a self-
%% sustaining SCS $C_t$ for every product $t \in \prodset(\strprofile)$.

\noindent ($\Rightarrow$) Suppose there exists a self-sustaining SCS
$C_t$ for a product $t$. Let $R$ be the set of nodes reachable from
$C_t$ which eventually can adopt product $t$. Formally, $R:=\bigcup_{m
  \in \nat} R_m$ where
\begin{itemize}
\item $R_0:=C_t$,
\item $R_{m+1}:=R_m \cup \{j \mid t \in \prodset(j) \mbox{ and } \sum\limits_{k \in \neighbour(j) \cap R_m} w_{kj} \geq \theta(j,t)\}$.
\end{itemize}

Let $s$ be the joint strategy such that for all $j \in R$, we have
$\strprofile_j=t$ and for all $k \in V \setminus R$, we have
$\strprofile_k=t_0$. It follows directly from the definition of $R$
that $\strprofile$ satisfies the following properties:
\begin{enumerate}
\item[(P1)] for all $i \in V$, $\strprofile_i=t_0$ or $\strprofile_i=t$,
\item[(P2)] for all $i \in V$, $\strprofile_i \neq t_0$ iff $i \in R$,
\item[(P3)] for all $i \in V$, if $i \in R$ then $\payoff_i(s) \geq 0$.
\end{enumerate}

We show that $\strprofile$ is a Nash equilibrium. Consider first any
$j$ such that $\strprofile_j=t$ (so $\strprofile_j\neq t_0$). By (P2)
$j \in R$ and by (P3) $\payoff_j(s)\geq 0$. Since
$\payoff_j(s_{-j},t_0)=0\leq \payoff_j(s)$, player $j$ does not gain
by deviating to $t_0$. Further, by (P1), for all $k \in
\neighbour(j)$, $\strprofile_k=t$ or $\strprofile_k=t_0$ and therefore
for all products $t' \neq t$ we have $\payoff_j(\strprofile_{-j},t')
<0 \leq \payoff_j(s)$. Thus player $j$ does not gain by deviating to
any product $t' \neq t$ either.

Next, consider any $j$ such that $\strprofile_j=t_0$. We have
$\payoff_j(\strprofile)=0$ and from (P2) it follows that $j \not\in
R$. By the definition of $R$ we have $\sum\limits_{k \in \neighbour(j)
  \cap R} w_{kj} < \theta(j,t)$. Thus
$\payoff_j(\strprofile_{-j},t) <0$. Moreover, for all products $t'
\neq t$ we also have $\payoff_j(\strprofile_{-j},t') <0$ for the same
reason as above. So player $j$ does not gain by a unilateral
deviation. We conclude that $\strprofile$ is a Nash equilibrium.
\end{proof}

Next, for a product $t \in \products$, we define the set
$X_t:=\bigcap_{m \in \nat} X_t^m$, where
\begin{itemize}
\item $X_t^0:=\{i \in V \mid t \in \prodset(i)\}$,
\item $X_t^{m+1}:=\{i \in V \mid \sum_{j \in \neighbour(i) \cap
  X_t^m} w_{ji} \geq \theta(i,t)\}$.
\end{itemize}

We need the following characterization result.

\begin{lemma}
\label{lm:proc-Ne-exists}
Let $\snet$ be a network
whose underlying graph has no source nodes. 
There exists a non-trivial Nash equilibrium in
$\mathcal{G}(\snet)$ iff there exists a product $t$ such that $X_t
\neq \emptyset$.
\end{lemma}

\begin{proof}
Suppose $\snet=(\sgraph,\products,\prodset,\theta)$.

\noindent
$(\Rightarrow)$ It follows directly from the definitions
that if there is a self-sustaining SCS $C_t$ for product $t$,
then $C_t \subseteq X_t$. Suppose now that for all $t$, $X_t =
\emptyset$. Then for all $t$, there is no self-sustaining SCS
for $t$. So by Lemma~\ref{lm:Ne-nosrc}, $\obar{t_0}$ is a unique
Nash equilibrium.

\noindent$(\Leftarrow)$ Suppose there exists $t$ such that $X_t \neq
\emptyset$. Let $\strprofile$ be the joint strategy defined as
follows:
\[
s_i:= \begin{cases}
      t & \mathrm{if}\ i \in X_t \\
      t_0 & \mathrm{if}\ i \not\in X_t
\end{cases}
\]

By the definition of $X_t$, for all $i \in X_t$,
$\payoff_i(\strprofile) \geq 0$. So no player $i \in X_t$ gains by
deviating to $t_0$ (as then his payoff would become $0$) or to a
product $t' \neq t$ (as then his payoff would become negative since no
player adopted $t'$). Also, by the definition of $X_t$ and of the
joint strategy $\strprofile$, for all $i \not\in X_t$ and for all $t' \in
\prodset(i)$, $\payoff_i(t',\strprofile_{-i})<0$. Therefore, no player
$i \not\in X_t$ gains by deviating to a product $t'$ either. It
follows that $\strprofile$ is a Nash equilibrium.
\end{proof}

This theorem leads to a direct proof of the claimed result.
\III

\NI \emph{Proof of Theorem~\ref{thm:nosource1}.}

\NI We use the following
procedure for checking for the existence of a non-trivial Nash
equilibrium. 
\[\mathsf{VerifyNash(\snet)}\]

$\mathit{found}:= \mathbf{false}$;

{\bf while} $\products \neq \emptyset$ {\bf and} $\neg \mathit{found}$ {\bf do}

\quad choose $t \in \products$;

\quad $\products := \products - \{t\}$;

\quad compute  $X_t$;

\quad $\mathit{found}:= (X_t \neq \emptyset)$

{\bf od}

{\bf return} $\mathit{found}$

\smallskip

On the account of Lemma~\ref{lm:proc-Ne-exists} this
procedure returns {\bf true} if a
non-trivial Nash equilibrium exists and {\bf false} otherwise.
To assess its complexity, note that for a network
$\snet=(\sgraph,\products,\prodset,\theta)$ and a fixed product $t$,
the set $X_t$ can be constructed in time $\bigo(n^3)$, where $n$
is the number of nodes in $\sgraph$. Indeed, each iteration of $X_t^m$
requires at most $\bigo(n^2)$ comparisons and the fixed point is
reached after at most $n$ steps. In the worst case, we need to compute
$X_t$ for every $t \in \products$, so the procedure runs in time
$\bigo(|\products| \cdot n^3)$.
\qed
\VV

In fact, the proof of Lemma~\ref{lm:proc-Ne-exists} shows that if a
non-trivial Nash equilibrium exists, then it can be constructed in
polynomial time as well. The complexity changes in the case of
determined Nash equilibria.

\begin{theorem} \label{thm:nosource2}
For a network $\snet$ whose underlying graph has no source
nodes, deciding whether the game $\mathcal{G}(\snet)$ has a determined Nash
equilibrium is NP-complete.

%% The problem of deciding whether for a network $\snet$ 
%% whose underlying graph has no source nodes
%% $\mathcal{G}(\snet)$ has a determined Nash equilibrium is NP-complete.
\end{theorem}

\begin{proof}
We modify the network given in Figure~\ref{fig:partition} so that 
the graph has no source nodes. To this end we `twin' each node $i \in \{1, \LL, n\}$
with a new node $i'$, also with the product set $\{t_1, t'_1\}$,
by adding edges $(i,i')$ and $(i',i)$. Additionally, we choose
the weights $w_{i i'}$ and $w_{i' i}$ and the corresponding thresholds so that when
$i$ and $i'$ adopt a common product, their payoff is positive.

For the so modified network we can now repeat the proof of Theorem~\ref{thm:nd}.
\end{proof}

\section{Finite best response property: general case}
\label{sec:FBRP}

We noted already that some social network games do not have a Nash
equilibrium.  A fortiori, such games do not have the finite best
response property (FBRP) defined in Subsection~\ref{subsec:games}. We
now analyze the complexity of determining whether a game has the FBRP.
We establish the following general result.

\begin{theorem} \label{thm:fbrp-hard}
Deciding whether for a network $\snet$ the game $\mathcal{G}(\snet)$
has the FBRP is co-NP-hard.
\end{theorem}

\begin{proof}
We prove that the complement of the problem is NP-hard.  We use again
an instance of the PARTITION problem in the form of $n$ positive
rational numbers $(a_1,\LL,a_n)$, appropriately normalized, so that
(this time) $\sum_{i=1}^n a_i = \frac12$, and the network given in
Figure~\ref{fig:partition-fbrp}.  For every node in the network, the
product set is $\{t_1, t_2\}$ and for $i\in\{1,\LL,n\}$ we set $w_{i
  a} = w_{i b} = a_i$.  The other three weights are $\frac12$.  Next,
we set $\theta(a, t_1) = \theta(b, t_2) = \frac34$ and stipulate that
for other inputs $\theta$ yields $\frac12$.

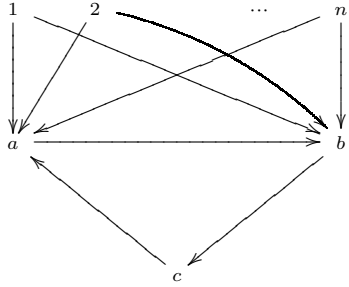
\begin{figure}[ht]
\centering
$
\def\objectstyle{\scriptstyle}
\def\labelstyle{\scriptstyle}
\xymatrix@W=10pt @R=40pt @C=15pt{
1 \ar[d]_{} \ar[rrrrd]_{}& 2 \ar[ld]^{} \ar@/^0.7pc/[rrrd]^{}& & \cdots &n \ar[lllld]^{} \ar[d]^{}\\
a  \ar[rrrr]^{}& & & &b  \ar[lld]^{}\\
& &c \ar[llu]^{}\\
}$
\caption{\label{fig:partition-fbrp}A network related to the FBRP}
\end{figure}

Suppose now that a solution to the considered instance of the
PARTITION problem exists, that is for some set $S \subseteq \{1, \LL,
n\}$ we have $\sum_{i\in S} a_i = \sum_{i\not\in S} a_i =
\frac{1}{4}$.  Consider the game $\mathcal{G}(\snet)$.  Assume now
that strategy $t_1$ is selected by each node $i \in S$ and strategy
$t_2$ by each node $i \in \{1, \LL, n\} \setminus S$.  Identify each
completion of this selection to a strategy profile with the triple of
strategies selected by the nodes $a, b$ and $c$.  Then it is easy to
check that Figure~\ref{fig:fbrp-infpath} exhibits an infinite best
response improvement path in $\mathcal{G}(\snet)$.

\begin{figure}[ht]
\centering
$
\def\objectstyle{\scriptstyle}
\def\labelstyle{\scriptstyle}
\xymatrix@W=8pt @C=15pt @R=15pt{
(\underline{t_1},t_1,t_2)\ar@{=>}[r]& (t_2, t_1, \underline{t_2})\ar@{=>}[r]& (t_2, \underline{t_1}, t_1)\ar@{=>}[d]\\
(t_1, \underline{t_2}, t_2)\ar@{=>}[u]& (t_1, t_2, \underline{t_1})\ar@{=>}[l]& (\underline{t_2}, t_2, t_1)\ar@{=>}[l]\\
}$
\caption{\label{fig:fbrp-infpath} An infinite best response improvement path}
\end{figure}

Suppose that an infinite best response improvement path $\xi$ exists
in $\mathcal{G}(\snet)$.  The payoff for the node $c$ is always at
most 0.  Hence $t_0$ is always a best response for $c$, so it is never
chosen by $c$ in $\xi$.  This means that in $\xi$, the node $c$
alternates between the strategies $t_1$ and $t_2$.  Hence strategies
$t_1$ and $t_2$ are also infinitely often selected in $\xi$ by the
nodes $a$ and $b$.  Consequently, the payoff for the node $a$ at the
moment it switches to $t_1$ has to be $\geq 0$ and the payoff for the
node $b$ at the moment it switches to $t_2$ has to be $\geq 0$.  This
means that $ \sum_{i \in \{1, \LL, n\} \mid s_i = t_1} w_{i a} \geq
\frac14 $ and $ \sum_{i \in \{1, \LL, n\} \mid s_i = t_2} w_{i b} \geq
\frac14.  $ But $\sum_{i=1}^n a_i =\frac12$ and for $i\in\{1,\LL,n\}$,
$w_{i a} = w_{i b} = a_i$, and $s_i \in \{t_1, t_2\}$.  So both above
inequalities are in fact equalities.  Consequently there exists a
solution to the considered instance of the PARTITION problem.
\end{proof}

It would be interesting to find out precisely the complexity of
determining whether a social network game has the FBRP. We conjecture
that it is not co-NP.

\section{Finite best response improvement property: special cases}
\label{sec:fbrp-special}

As in the case of Nash equilibria we now analyze the FBRP for social
network games whose underlying graph satisfies certain properties.

\subsection{Directed acyclic graphs}
\label{sec:fbrp-DAG}
We begin with social network games whose underlying graph is a DAG.
Then the following positive result holds.

\begin{theorem}
\label{thm:fbrp-DAG}
Consider a network $\snet$ whose underlying graph is a DAG.
Then the game $\mathcal{G}(\snet)$ has the FBRP.
\end{theorem}

\noindent This is a direct consequence of a stronger result, Theorem
\ref{thm:fip-DAG} of Section \ref{sec:fip-scc}.

\subsection{Simple cycles}

The property that the game has the FBRP does not hold anymore when the
underlying graph is a simple cycle.  To see this consider
Figure~\ref{fig:br0}(a). Suppose that $\bar{t}$ is a Nash equilibrium
in which each player gets a strictly positive payoff.
Figure~\ref{fig:br0}(b) then shows an infinite best response
improvement path. In each strategy profile, we underline the strategy
that is not a best response to the choice of other players. 

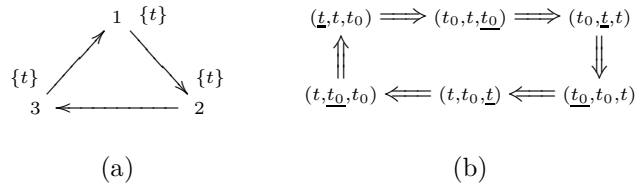
\begin{figure}[ht]
\centering
\begin{tabular}{ccc}
$
\def\objectstyle{\scriptstyle}
\def\labelstyle{\scriptstyle}
\xymatrix@W=10pt @C=15pt{
 &1 \ar[rd]^{} \ar@{}[rd]^<{\{t\}} \ar@{}[rd]^>{\{t\}}\\
3  \ar[ru]_{} \ar@{}[ru]^<{\{t\}} & &2 \ar[ll]^{}\\
}$
%% \caption{\label{fig:br1} A network with an infinite improvement path}
%% \end{figure}
&
&
%% \begin{figure}[ht]
%% \centering
$
\def\objectstyle{\scriptstyle}
\def\labelstyle{\scriptstyle}
\xymatrix@W=8pt @C=15pt @R=15pt{
(\underline{t},t, t_0)\ar@{=>}[r]& (t_0, t, \underline{t_0})\ar@{=>}[r]& (t_0, \underline{t}, t)\ar@{=>}[d]\\
(t, \underline{t_0}, t_0)\ar@{=>}[u]& (t, t_0, \underline{t})\ar@{=>}[l]& (\underline{t_0}, t_0, t)\ar@{=>}[l]\\
}$
%% \caption{\label{fig:br2} An infinite improvement path}
\\
\\
(a) & & (b)\\
\end{tabular}
\caption{\label{fig:br0} A network with an infinite best response improvement path}
\end{figure} 

Next consider Figure~\ref{fig:br1}(a). Suppose that both
$\overline{t_1}$ and $\overline{t_2}$ are Nash equilibria.  Then
Figure~\ref{fig:br1}(b) shows an infinite best response improvement
path. This essentially repeats the reasoning used to show that the
network given in Figure~\ref{fig:partition-fbrp} has an infinite best
response improvement path.
%% Note that also here at each step of this improvement path a
%% best response is used.

\begin{figure}[ht]
\centering
\begin{tabular}{ccc}
$
\def\objectstyle{\scriptstyle}
\def\labelstyle{\scriptstyle}
\xymatrix@W=10pt @C=15pt{
 &1 \ar[rd]^{} \ar@{}[rd]^<{\{t_1,t_2\}} \ar@{}[rd]^>{\{t_1,t_2\}}\\
3  \ar[ru]_{} \ar@{}[ru]^<{\{t_1,t_2\}} & &2 \ar[ll]^{}\\
}$
%% \caption{\label{fig:br1} A network with an infinite improvement path}
%% \end{figure}
&
&
$
\def\objectstyle{\scriptstyle}
\def\labelstyle{\scriptstyle}
\xymatrix@W=8pt @C=15pt @R=15pt{
(\underline{t_1},t_1,t_2)\ar@{=>}[r]& (t_2, t_1, \underline{t_2})\ar@{=>}[r]& (t_2, \underline{t_1}, t_1)\ar@{=>}[d]\\
(t_1, \underline{t_2}, t_2)\ar@{=>}[u]& (t_1, t_2, \underline{t_1})\ar@{=>}[l]& (\underline{t_2}, t_2, t_1)\ar@{=>}[l]\\
}$
\\
\\
(a) & &(b)\\
\end{tabular}
\caption{\label{fig:br1} Another network with an infinite best response improvement path}
\end{figure}
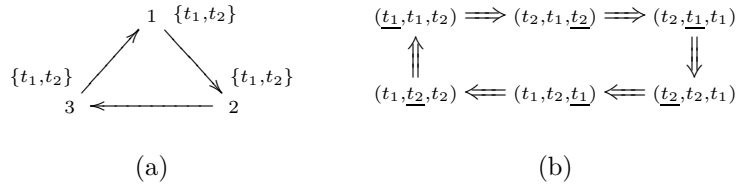 

One can easily generalize the above two examples to simple cycles with
more than three nodes.  Below we show that when the game does not have
the FBRP, the network is necessarily of one of the above two
types.  This will allow us to establish the following complexity
result concerning the FBRP property.

\begin{theorem} 
\label{thm:fbrp-cycle}
Consider a network $\snet=(\sgraph,\products,\prodset,\theta)$ such
that $G$ is a simple cycle.  There is a procedure that runs in time
$\bigo(|\products| \cdot n)$, where $n$ is the number of nodes
in $\sgraph$, that decides whether the game $\mathcal{G}(\snet)$ has
the FBRP.
\end{theorem}

We need the following characterization result.

\begin{theorem}
\label{thm:fbrp-cycle3}
Let $\snet$ be a network whose underlying graph is a simple cycle.

\begin{enumerate}[(i)]
\item Suppose that $\snet$ has 2 nodes. Then 
the game $\mathcal{G}(\snet)$ has the FBRP.

\item Suppose that $\snet$ has at least 3 nodes. Then the game
  $\mathcal{G}(\snet)$ does not have the FBRP iff either it has a
  determined Nash equilibrium $s$ such that for all $i$, $p_i(s) > 0$
  or it has two determined Nash equilibria.
\end{enumerate}
\end{theorem}

\begin{proof}

\NI
$(i)$ A simple analysis, which we leave to the reader, shows that the longest
possible improvement path is of length five and is of the form
$(t_1, t_2)$, $(t_0, t_2)$, $(t_2, t_2)$, $(t_2, t_0)$, $(t_0, t_0)$.
\II

\NI
$(ii)$ $(\Ra)$ 
Consider an infinite best response improvement path $\xi$.  Some node
changes his strategy in $\xi$ infinitely often.  This means that some
node, say $i$, selects in $\xi$ some product $t$ infinitely often.
Indeed, otherwise from some moment on in each strategy profile in
$\xi$ its strategy would be $t_0$, which is not the case.

Each time node $i$ switches in $\xi$ to the product $t$, it selects a
best response, so its payoff becomes at least $0$. Consequently, at
the moment of such a switch its predecessor $i \ominus 1$'s strategy is
necessarily $t$ as well.  So if from some moment on node $i \ominus 1$
does not switch from the strategy $t$, then node $i$ does not switch
from $t$ either.  This shows that node $i \ominus 1$ also selects
in $\xi$ product $t$ infinitely often.  Iterating this reasoning we
conclude that each node selects in $\xi$ the product $t$ infinitely
often.
Therefore, for all $i$, $t \in \prodset(i)$. Since the payoff of $i$
depends only on the choice of $i \ominus 1$, we also have that
$\payoff_i(\obar{t}) \geq 0$ for all $i$. By Theorem~\ref{thm:cycle},
$\bar{t}$ is a Nash equilibrium.

This shows that if a node selects in $\xi$ some product $t_1$
infinitely often, then all nodes select in $\xi$ the product $t_1$
infinitely often and $\overline{t_1}$ is a Nash equilibrium.  Suppose
now that $\bar{t}$ is a unique determined Nash equilibrium.  This
means that all other products are selected in $\xi$ finitely often.
So from some moment on in $\xi$ nodes select only $t$ or $t_0$.  In
this suffix $\eta$ of $\xi$ each node selects $t$ infinitely often.
Further, each switch to $t$ from $t_0$ is a better response. Hence
each time a node switches in $\eta$ to $t$ its payoff becomes $> 0$.
This shows that $\bar{t}$ is a determined Nash equilibrium such that
for all $i$, $p_i(\bar{t}) > 0$.
\II

\NI
$(\La)$ Suppose first that the game $\mathcal{G}(\snet)$ has a determined
Nash equilibrium $s$ such that for all $i$, $p_i(s) > 0$.
By Theorem~\ref{thm:cycle} $s$ is of the form $\bar{t}$ for some product $t$.
Then consider the following strategy profile:
\[
s := (t, \LL, t, t_0).
\]
First schedule  node 1 that has a better response, namely $t_0$.  Next,
schedule node $n$ for which $t$ is a best response.  After these two
steps the strategy profile becomes a rotation of $s$ by $1$.
Iterating this selection procedure we obtain an infinite best response
improvement path.

Next, suppose that the game $\mathcal{G}(\snet)$ has two determined
Nash equilibria. By 
Theorem~\ref{thm:cycle} they are of the form $\overline{t_1}$ and $\overline{t_2}$
for some products $t_1$ and $t_2$. 
Then consider the following strategy profile:
\[
s := (t_1, \LL, t_1, t_2).
\]
First schedule node 1 for which $t_2$ is a best response.  Next,
schedule node $n$ for which $t_1$ is a best response.  After these two
steps the strategy profile becomes a rotation of $s$ by $1$.
Iterating this selection procedure we obtain an infinite best response
improvement path.
\end{proof}

This theorem leads to a direct proof of the claimed result.
\III

\NI
\emph{Proof of Theorem~\ref{thm:fbrp-cycle}.}

Thanks to the above theorem, one can check whether
$\mathcal{G}(\snet)$ has the FBRP as follows. First apply the
procedure $\mathsf{VerifyNashCycle(\snet)}$ defined in the proof of
Theorem~\ref{thm:cycle-NE}, appropriately modified to check for the
existence of a Nash equilibrium with strictly positive payoffs. If
such an equilibrium does not exist, then we use a modified version of
$\mathsf{VerifyNashCycle(\snet)}$ to check whether
$\mathcal{G}(\snet)$ has two determined Nash equilibria.

\subsection{Graphs with no source nodes}
As in Section \ref{sec:Ne-special} we conclude by considering social
networks whose underlying graphs have no source nodes. The following
counterpart of Theorem~\ref{thm:fbrp-hard} holds.

\begin{theorem} \label{thm:fbrp-hard-nosource}
For a network $\snet$ whose underlying graph has no source nodes,
deciding whether the game $\mathcal{G}(\snet)$ has the FBRP is
co-NP-hard.
\end{theorem}

\begin{proof}
The proof is a modification of the proof of Theorem~\ref{thm:fbrp-hard}.  Given an
instance of the PARTITION problem we use the following modification of
the network given in Figure~\ref{fig:partition-fbrp}.  We `twin' each
node $i \in \{1, \ldots, n\}$ with a new node $i'$, also with the
product set $\{t_1, t_2\}$, by adding the edges $(i,i')$ and
$(i',i)$. We choose the weights $w_{i i'}, w_{i' i}$, where $i \in
\{1, \ldots, n\}$ and the thresholds so that when $i$ and $i'$ adopt a
common product, their payoffs are positive.

Suppose that a solution to the considered instance of the PARTITION
problem exists. Then we extend the joint strategy considered in the
proof of Theorem~\ref{thm:fbrp-hard} by additionally assigning $t_1$
to each node $i'$ such that $i \in S$ and $t_2$ to each node $i'$ such
that $i \in \{1, \ldots, n\} \setminus S$. Then, as before, there is
no finite best response improvement path starting in this joint
strategy, so $\mathcal{G}(\snet)$ does not have the FBRP.

Suppose that an infinite best response improvement path $\xi$ exists
in $\mathcal{G}(\snet)$. By Theorem \ref{thm:fbrp-cycle3}(i) each
network consisting of just the twined nodes $(1,1'),\ldots,(n,n')$ has
the FBRP. Therefore there exists a joint strategy $\strprofile^k$ in
$\xi$ such that in the (infinite) suffix of $\xi$ starting at
$\strprofile^k$ (call it $\xi_k$), none of the following nodes are
scheduled: $1,1',\ldots,n,n'$. The payoff for node $c$ in $\xi_k$ is
at most 0.  Hence $t_0$ is always a best response for it.  If $c$
chooses $t_0$ in $\xi$ then clearly $\xi$ is finite which contradicts
the assumption. Also, if there exists a $\strprofile^m$ in $\xi_k$
such that $c$ never changes its strategy in the suffix of $\xi_k$
starting at $\strprofile_m$, then again $\xi$ is finite. This means
that in $\xi_k$ the node $c$ alternates between the strategies $t_1$
and $t_2$.  Hence in $\xi_k$ strategies $t_1$ and $t_2$ are also
infinitely often selected by the nodes $a$ and $b$.  Consequently, the
payoff for the node $a$ at the moment it switches to $t_1$ has to be
$\geq 0$ and the payoff for the node $b$ at the moment it switches to
$t_2$ has to be $\geq 0$.  This means that $ \sum_{i \in \{1, \LL, n\}
  \mid s_i = t_1} w_{i a} \geq \frac14 $ and $ \sum_{i \in \{1, \LL,
  n\} \mid s_i = t_2} w_{i b} \geq \frac14.  $ But $\sum_{i=1}^n a_i
=\frac12$ and for $i\in\{1,\LL,n\}$, $w_{i a} = w_{i b} = a_i$, and
$s_i \in \{t_1, t_2\}$.  So both above inequalities are in fact
equalities.  Consequently there exists a solution to the considered
instance of the PARTITION problem.
\end{proof}

\section{Finite improvement property: general case}
\label{sec:FIP}

In the previous section we noted that some social network games do not
have the FBRP. A fortiori, such games do not have the FIP either. We
now analyze the complexity of determining whether a game has the
FIP. We start with the following analogue of
Theorem~\ref{thm:fbrp-hard}, though the proof is more complex.

\begin{figure}[ht]
\centering
$
\def\objectstyle{\scriptstyle}
\def\labelstyle{\scriptstyle}
\xymatrix@W=10pt @R=25pt @C=15pt{
1 \ar[d]_{} \ar[rrrrd]_{}& 2 \ar[ld]^{} \ar@/^0.7pc/[rrrd]^{}& & \cdots &n \ar[lllld]^{} \ar[d]^{}\\
a  \ar[rrrr]_{\frac34} \ar@{}[rrd]_<{\{t_2,t_3\}}& & & &b  \ar[lld]^{\frac34} \ar@{}[lld]^<{\{t_1,t_2\}}\\
& &c \ar[llu]^{\frac34} \ar@{}[d]_<{\{t_1,t_3\}}\\
& &\{t_3\} \ar[u]_{\frac12}
}$
\caption{\label{fig:partition-fip}A network related to the FIP}
\end{figure}
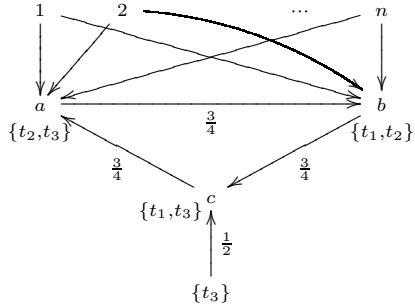

\begin{theorem} \label{thm:fip-hard}
Deciding whether for a network $\snet$ the game
$\mathcal{G}(\snet)$ has the FIP is co-NP-hard.
\end{theorem}

\begin{proof}
We prove that the complement of the problem is NP-hard.  We use again
an instance of the PARTITION problem in the form of $n$ positive
rational numbers $(a_1,\LL,a_n)$, appropriately normalised, so that
(this time) $\sum_{i=1}^n a_i = \frac12$, and the network given in
Figure~\ref{fig:partition-fip}.  For each node $i\in\{1,\LL,n\}$ we
set the product set to be $\{t_1,t_2\}$ and $w_{i a} = w_{i b} = a_i$.
The weights of the other edges are shown in the figure. 

Since for all $i \in \{1,\ldots,n\}$, $a_i$ is rational, it has the
form $a_i = \frac{l_i}{r_i}$. Let $\tau=\frac{1}{4\cdot r_1 \cdot
  \ldots \cdot r_n}$. 
%% Clearly every $a_i$ and $\frac14$ is a multiple
%% of $\tau$.
The following property holds.

\begin{property}
\label{partition-tau}
Given an instance $(a_1,\ldots,a_n)$ of the PARTITION problem and
$\tau$ defined as above, for all $S \subseteq \{1,\ldots,n\}$
\begin{enumerate}[(i)]
\item if $\sum_{i \in S} a_i < \frac14$, then $\sum_{i \in S} a_i \leq \frac14 -\tau$,
\item if $\sum_{i \in S} a_i > \frac14$, then $\sum_{i \in S} a_i \geq \frac14 +\tau$.
\end{enumerate}
\end{property}
\begin{proof}
By definition, each $a_i$ and $\frac14$ is a multiple of $\tau$. Thus
$\sum_{i \in S} a_i = x \cdot \tau$ and $\frac14 = y \cdot \tau$ where
$x$ and $y$ are integers. \\
\noindent {\it (i)} If $x \cdot \tau < y \cdot \tau$, then $x \cdot
\tau \leq (y-1) \cdot \tau$. Therefore $\sum_{i \in S} a_i \leq \frac14 -
\tau$.\\
\noindent The proof of {\it (ii)} is analogous.
\end{proof}
Note that given $(a_1,\ldots,a_n)$, $\tau$ can be defined in
polynomial time. Let the thresholds be defined as follows:
$\theta(a,t_2)=\frac12$, $\theta(a,t_3)=\frac14+\tau$,
$\theta(b,t_1)=\frac12$, $\theta(b,t_2)=\frac12 + \tau$ and for node
$c$, $\theta(c,t_1)=\theta(c,t_3) = \frac14$.

Suppose now that a solution to the considered instance of the
PARTITION problem exists. That is, for some set $S \subseteq \{1, \LL,
n\}$ we have $\sum_{i\in S} a_i = \sum_{i\not\in S} a_i =
\frac{1}{4}$.  Consider the game $\mathcal{G}(\snet)$.  Assume now
that strategy $t_1$ is selected by each node $i \in S$, strategy $t_2$
by each node $i \in \{1, \LL, n\} \setminus S$ and that the node
labelled by $\{t_3\}$ selected strategy $t_3$.  Identify each
completion of this selection to a joint strategy with the triple of
strategies selected by the nodes $a, b$ and $c$.  Then it is easy to
check that Figure~\ref{fig:FIP-hardness2}(a) exhibits an infinite improvement path
in $\mathcal{G}(\snet)$; for convenience we list the corresponding
payoffs in Figure~\ref{fig:FIP-hardness2}(b).

\begin{figure}[ht]
\centering
\begin{tabular}{ccc}
$
\def\objectstyle{\scriptstyle}
\def\labelstyle{\scriptstyle}
\xymatrix@W=8pt @C=10pt @R=15pt{
(\underline{t_2},t_2,t_3)\ar@{=>}[r]& (t_3, \underline{t_2}, t_3)\ar@{=>}[r]& (t_3, t_1, \underline{t_3})\ar@{=>}[d]\\
(t_2, t_2, \underline{t_1})\ar@{=>}[u]& (t_2, \underline{t_1}, t_1) \ar@{=>}[l]& (\underline{t_3}, t_1, t_1)\ar@{=>}[l]\\
}$
&
$
\def\objectstyle{\scriptstyle}
\def\labelstyle{\scriptstyle}
\xymatrix@W=8pt @C=10pt @R=15pt{
(-\frac14,\frac12-\tau,\frac14)\ar@{=>}[r]& (\frac12-\tau, -\frac14-\tau, \frac14)\ar@{=>}[r]& (\frac12-\tau, -\frac14, \frac14)\ar@{=>}[d]\\
(-\frac14, \frac12 -\tau, -\frac14)\ar@{=>}[u]& (-\frac14, -\frac14, -\frac14) \ar@{=>}[l]& (-\frac14-\tau, -\frac14, \frac12)\ar@{=>}[l]\\
}$
\\
\\
(a) & (b)\\
\end{tabular}

\caption{\label{fig:FIP-hardness2} An infinite improvement path and the corresponding payoffs}
\end{figure}

Conversely, we show that if there is no solution to the considered
instance of the PARTITION problem, then the game $\mathcal{G}(\snet)$
has the FIP. Consider an improvement path $\xi$. Let $k\geq 0$ be the
first index (of a joint strategy) in $\xi$ such that in the suffix of
$\xi$ starting at $k$, none of the source nodes is anymore
scheduled. Let $\setn{n}$ denote the set $\{1,2,\ldots,n\}$,
$\strprofile^k$ the $k$-th element of $\xi$, $\val(t_1) = \sum_{i \in
  \setn{n} | \strprofile_i^k=t_1} a_i$ and $\val(t_2) = \sum_{i \in
  \setn{n} | \strprofile_i^k=t_2} a_i$. There are two cases.

\II \NI \emph{Case 1:} $\val(t_2) < \frac14$. By Property
\ref{partition-tau}, it follows that $\val(t_2) \leq \frac14
-\tau$. We argue that $t_2$ is never a better response for node
$a$. Suppose $\strprofile_a^k=t_3$. We have the following two
possibilities:
\begin{itemize}
\item $\strprofile^k_c=t_3$; then
  $\payoff_a(\strprofile^k)=\frac12-\tau$, so $t_3$ is node $a$'s best
  response,
\item $\strprofile^k_c \neq t_3$; then
  $\payoff_a(\strprofile^k)=-\frac14 - \tau$. If node $a$ switches to
  $t_2$, then the payoff is $\payoff_a(s^k_{-a},t_2) \leq (\frac14 -
  \tau) - \frac12 = -\frac14 - \tau$. Thus $t_2$ is not a better
  response for $a$.
\end{itemize}

Similarly, it is easy to check that if $\strprofile_a^k=t_0$, then
$t_2$ is not a better response for node $a$. If node $a$ never changes
its strategy in $\xi$ then clearly $\xi$ is finite. If it does, then
let $k_a >k$ be the first index in $\xi$ such that node $a$ changes
its strategy (i.e., $\strprofile_a^k \neq \strprofile_a^{k_a}$). Since
$t_2$ is never a better response for $a$, in the suffix of $\xi$ (call
it $\xi_a$) starting at $k_a$, node $a$ never chooses $t_2$. In
$\xi_a$ the better response of node $b$ can be $t_1$, $t_2$ or
$t_0$. However, since node $a$ never chooses $t_2$, $b$'s payoff in
$\xi_a$ depends only on the choice made by the source nodes (which
never changes in $\xi_a$). Therefore, there exists an index $k_b$ in
$\xi$, where $k<k_a<k_b$, such that in the suffix of $\xi$ starting at
$k_b$ (call it $\xi_b$), node $b$ never changes its strategy. The
payoff of node $c$ depends on the choice made by nodes $b$ and the
source node marked $\{t_3\}$. Since in $\xi_b$ node $b$ never changes
its strategy, there is a suffix of $\xi_b$ in which node $c$ never
changes its strategy. Therefore $\xi_b$, and hence $\xi$ is finite.

\II 

\NI 
\emph{Case 2:} $\val(t_2) > \frac14$. By Property
\ref{partition-tau}, it follows that $\val(t_2) \geq \frac14 +
\tau$. We argue that $t_1$ is never a better response for node
$b$. Suppose $\strprofile_b^k=t_2$. We have the following two
possibilities:

\begin{itemize}
\item $\strprofile^k_a=t_2$; then $\payoff_b(\strprofile^k) \geq
  \frac12$, so $t_2$ is node $b$'s best response,

\item if $\strprofile^k_a \neq t_2$, then $\payoff_b(\strprofile^k)
  \geq (\frac14+\tau)-\frac12-\tau = -\frac14$. If node $b$ switches
  to $t_1$, then since $\val(t_1) < \frac14$, the payoff is
  $\payoff_b(s^k_{-b},t_1) < \frac14 - \frac12 = -\frac14$. Thus $t_1$
  is not a better response for $b$.
\end{itemize}

Similarly, it is easy to check that if $\strprofile_b^k=t_0$ then
$t_1$ is not a better response for node $b$. This implies that there
exists an index $k_b >k$ in $\xi$ such that for the suffix of $\xi$
(call it $\xi_b$) starting at $k_b$, node $b$ never chooses $t_1$. In
$\xi_b$, since node $b$ never chooses $t_1$, the better response of
node $c$ can be either $t_3$ or $t_0$. Thus $c$'s payoff in $\xi_b$
depends only on the choice made by the source node labelled by
$\{t_3\}$. Therefore, there exists an index $k_c$ in $\xi$, where
$k<k_b<k_c$, such that in the suffix of $\xi$ starting at $k_c$ (call
it $\xi_c$), node $c$ is never scheduled. It follows that $\xi_c$ and
hence $\xi$ is finite.

Since the choice of the improvement path was arbitrary, it follows
that $\mathcal{G}(\snet)$ has the FIP.
\end{proof}

\section{Finite improvement property: special cases}
\label{sec:fip-special}

In this section we clarify whether the special classes of social
network games have the FIP and if not, what the complexity of the
problem is.  

First, we would like to mention the following question to which
we did not succeed to find an answer.
\II

\NI
\textbf{Open problem:}
Consider a network $\snet$ whose underlying graph is a simple cycle.
What is the complexity of determining 
whether $\mathcal{G}(\snet)$ has the FIP?
\II

In what follows we make use of the following simple observation.
\begin{note} \label{not:inf}
 Consider a game $\mathcal{G}(\snet)$. If a node $i$ is infinitely
 often selected in an improvement path, then so is a node $j \in
 \neighbour(i)$.  \HB
\end{note}

\subsection{Graphs with special strongly connected components}
\label{sec:fip-scc}
We begin with the strategic games associated with
the networks whose underlying graph is a DAG. 
Then the following positive result holds.

\begin{theorem}
\label{thm:fip-DAG}
Consider a network $\snet$ whose underlying graph is a DAG.
Then the game $\mathcal{G}(\snet)$ has the FIP.
\end{theorem}

\begin{proof}
Suppose not. Then by repeatedly using Note~\ref{not:inf} we obtain an
infinite path in the underlying graph. This yields a contradiction.
\end{proof}

We now generalize this result to a larger class of directed graphs.
First we consider the case of two player games.

\begin{theorem}
\label{thm:2-FIP}
Every two player social network game has the FIP.
\end{theorem}

\begin{proof}
By Theorem~\ref{thm:fip-DAG} we can assume that the underlying graph is a 
cycle, say $1 \to 2 \to 1$. Consider an improvement path $\rho$. 
By removing, if necessary, some steps we can assume that the players 
alternate their moves in $\rho$.

In what follows given an element of $\rho$ (that is not the last one)
we underline the strategy of the player who moves, i.e., selects a
better response.  We call each element of $\rho$ of the type
$(\underline{t}, t)$ or $(t, \underline{t})$ a \emph{match} and use
$\Rightarrow$ to denote the transition between two consecutive joint
strategies in $\rho$.  Further, we shorten the statement
``each time player $i$ switches his strategy his payoff strictly
increases and it never decreases when his opponent switches strategy''
to ``player $i$'s payoff steadily goes up''.

Consider now two successive matches in $\rho$, based respectively on
the strategies $t$ and $t_1$. The corresponding segment of $\rho$
is one of the following four types.

\NI
\emph{Type 1}. $(\underline{t}, t) \Rightarrow^{*} (\underline{t_1}, t_1)$.

The fragment of $\rho$ that starts at 
$(\underline{t}, t)$ and finishes at $(\underline{t_1}, t_1)$ has the following form:
\[
(\underline{t}, t) \Rightarrow (t_2, \underline{t}) \Rightarrow^{*} (t_1, \underline{t_3}) \Rightarrow
(\underline{t_1}, t_1).
\]

Note that player $1$'s payoff can drop in a segment of $\rho$ only if this segment
contains a transition of the form $(t', \underline{t'}) \Rightarrow (\underline{t'}, t_1)$.
So in the considered segment player $1$'s payoff steadily goes up.  Additionally, in the step $(t_1,
\underline{t_3}) \Rightarrow (\underline{t_1}, t_1)$ his payoff increases by
$w_{2 1}$.

In turn, in the step $(\underline{t}, t) \Rightarrow (t_2,
\underline{t})$ player $2$'s payoff decreases by $w_{1 2}$ and in the
remaining steps his payoff steadily goes up.
So $p_1(\bar{t}) + w_{2 1} < p_1(\overline{t_1})$ and $p_2(\bar{t}) - w_{1 2} < p_2(\overline{t_1})$.
\II

\NI
\emph{Type 2}. $(\underline{t}, t) \Rightarrow^{*} (t_1, \underline{t_1})$.

For the analogous reason as above player $1$'s payoff steadily goes up.  In
turn, in the first step of $(\underline{t}, t) \Rightarrow^{*} (t_1,
\underline{t_1})$ the payoff of player 2 decreases by $w_{1 2}$, while
in the last step (in which player 1 moves) his payoff increases by
$w_{1 2}$.  So these two payoff changes cancel each other.
Additionally, in the remaining steps player $2$'s payoff steadily goes
up.  So $p_1(\bar{t}) < p_1(\overline{t_1})$ and $p_2(\bar{t}) <
p_2(\overline{t_1})$.  \II

\NI
\emph{Type 3}. $(t, \underline{t}) \Rightarrow^{*} (\underline{t_1}, t_1)$.

This type is symmetric to Type 2, so
$p_1(\bar{t}) < p_1(\overline{t_1})$ and $p_2(\bar{t}) < p_2(\overline{t_1})$.
\II

\NI
\emph{Type 4}. $(t, \underline{t}) \Rightarrow^{*} (t_1, \underline{t_1})$.

This type is symmetric to Type 1, so
$p_1(\bar{t}) - w_{2 1} < p_1(\overline{t_1})$ and $p_2(\bar{t}) + w_{1 2} < p_2(\overline{t_1})$.
\II

We summarize in Table~\ref{tab:T} the changes in the payoffs $p_1$ and
$p_2$ between the considered two matches.

\begin{table}[htbp]
  \begin{center}
    \leavevmode
\begin{tabular}{|l|l|l|}
\hline 
Type & $p_1$ & $p_2$ \\
\hline 
\hline 
1    & increases      & decreases \\
     & by $> w_{2 1}$ & by $< w_{1 2}$ \\
\hline
2, 3 & increases      & increases \\
\hline
4    & decreases      & increases \\
     & by $< w_{2 1}$ & by $> w_{1 2}$ \\
\hline
\end{tabular}
  \end{center}
    \caption{\label{tab:T}Changes in $p_1$ and $p_2$}
\end{table}

Consider now a match $(\underline{t}, t)$ in $\rho$ and a match $(\underline{t_1}, t_1)$
that appears later in $\rho$.
Let $T_i$ denote the number of internal segments of type $i$ that occur in the fragment of $\rho$
that starts with $(\underline{t}, t)$ and ends with $(\underline{t_1}, t_1)$.
\II

\NI
\emph{Case 1}. $T_1 \geq T_4$.

Then Table~\ref{tab:T} shows that the aggregate increase in $p_1$ in segments of type 1
exceeds the aggregate decrease in segments of type 4.
So $p_1(\bar{t}) < p_1(\overline{t_1})$.
\II

\NI
\emph{Case 2}. $T_1 < T_4$.

Then analogously Table~\ref{tab:T} shows that 
$p_2(\bar{t}) < p_2(\overline{t_1})$.
\II

We conclude that $t \neq t_1$. By symmetry the same conclusion holds
if the considered matches are of the form $(t, \underline{t})$ and
$(t_1, \underline{t_1})$.  This proves that each match occurs in
$\rho$ at most once.  So in some suffix $\eta$ of $\rho$ no match occurs.
But each step in $\eta$ increases the social welfare, so $\eta$ is
finite, and consequently $\rho$ is.
\end{proof}

In social network games the players share at least one strategy,
$t_0$, that ensures each of them the zero payoff. Also, the weights and thresholds
are drawn from specific intervals. However, these properties
are not used in the above proof. As a result the above
proof shows that each of the following two player games has
the FIP.

\begin{itemize}
\item The set of strategies of player $i$ is a finite set $S_i$,

\item the payoff function is defined by $p_i(s) := f_i(s_i) + a_i  (s_i = s_{-i})$,
where $f_i: S_i \to \mathbb{R}$, $a_i > 0$ and $(s_i = s_{-i})$ is defined by
\[
(s_i = s_{-i}) := \begin{cases}
        1     & \mathrm{if}\  s_i = s_{-i} \\
        0     & \mathrm{otherwise}
        \end{cases}
\]
\end{itemize}
Intuitively, $a_i$ can be viewed as a bonus for player $i$ for coordinating with his opponent.

It is worthwhile to notice that if we allowed that the weights depend on the product
and modified the definition of the payoff function
to
\[
\sum\limits_{j \in \inflset_i^t(\strprofile)} w_{ji}(t) - \theta(i,t),
\]
then the above result would not hold.

\begin{example}
\label{exa:12}
\rm

Suppose that the product sets of player 1 and player 2 are $\{t_1,
t_2, t_3\}$ and $\{t_1, t_3, t_4\}$ respectively. Define the threshold
functions as follows:
\[
\begin{array}{c}
\theta(1, t_1) = 0.9, \  \theta(1, t_2) = 0.8, \ \theta(1, t_3) = 1, \\
\theta(2, t_1) = 1, \ \theta(2, t_3) = 0.9, \ \theta(2, t_4) = 0.8,
\end{array}
\]
and suppose the parametrized weight functions are such that
\[
w_{2 1}(t_1) = 0, \ w_{2 1}(t_3) = 0.3, \ w_{1 2}(t_1) = 0.3, \ w_{1 2}(t_3) = 0.
\]

Then it is easy to check that the iteration of the following sequence forms an
infinite improvement path:
\[
(\underline{t_1}, t_1) \Rightarrow (t_2, \underline{t_1}) \Rightarrow 
(\underline{t_2}, t_3) \Rightarrow (t_3, \underline{t_3}) \Rightarrow 
(\underline{t_3}, t_4) \Rightarrow (t_1, \underline{t_4}) \Rightarrow 
(\underline{t_1}, t_1).
\]
\HB
\end{example}

We can now draw a conclusion about a larger class of 
social network games.

\begin{theorem}
\label{thm:FIP-SCC2}
Consider a network $\snet$ such that each strongly connected component of
the underlying graph is a cycle of length 2.
Then the game $\mathcal{G}(\snet)$ has the FIP.
\end{theorem}

\begin{proof}
  Suppose $\snet =(\sgraph,\products,\prodset,\theta)$. Consider the
  condensation of $\sgraph$, i.e., the DAG $\sgraph'$ resulting from
  contracting each cycle of $\sgraph$ to a single node.  We enumerate
  the nodes of $\sgraph$ using the function $\dagrank()$ (introduced
  in Subsection~\ref{subsec:dag}) such that if $\dagrank(i) < \dagrank(j)$, then there is no path in
  $\sgraph$ from $j$ to $i$, 
and subsequently modify it to an
  enumeration $\dagrank'()$ of the nodes of $\sgraph'$, by replacing
  each contraction of a cycle by its two nodes.

%% \begin{equation}
%%  \label{equ:rankrep}
%% \mbox{if $\dagrank(i) < \dagrank(j)$, then there is no path in
%%   $\sgraph$ from $j$ to $i$.}
%% \end{equation}
  
  Suppose now that an infinite improvement path $\rho$ in
  $\mathcal{G}(\snet)$ exists.  Choose the first node $i$ in the
  enumeration $\dagrank'()$ that is infinitely often selected in
  $\rho$. By Note~\ref{not:inf} $i$ lies on a cycle, say $i \to i'\to
  i$, in $\sgraph$.  Consider a suffix of $\rho$ in which the nodes
  that precede $i$ in $\dagrank'()$ do not appear anymore.  In
  particular, by the choice of $i$, the neighbours of $i$ or $i'$
  precede $i$ in $\dagrank'()$, so they do not appear in this suffix
  either.  Delete from each element of this suffix the strategies of
  the nodes that differ from $i$ and $i'$. We obtain in this way an
  infinite improvement path in the two player game $G'$ associated
  with the weighted directed graph related to the cycle $i \to i'\to
  i$, which contradicts Theorem~\ref{thm:2-FIP}.
  
  There is a small subtlety in the last step that we should clarify.
  The payoff functions in the game $G'$ need to take into account the
  weights of the edges from all the neighbours of $i$ and $i'$ and the
  final strategies chosen by these neighbours. For instance, in the
  case of the directed graph from Figure~\ref{fig:3} the game $G'$ has
  the players $i$ and $i'$ but the weights of the edges $k \to i$ and
  $k \to i'$ and the final strategy of node $k$ need to be taken into
  account in the computation of the payoff functions for,
  respectively, nodes $i$ and $i'$.

\begin{figure}[ht]
\centering
$
\def\objectstyle{\scriptstyle}
\def\labelstyle{\scriptstyle}
\xymatrix@W=10pt @R=25pt @C=15pt{
& k \ar[dl]_{} \ar[dr]_{}\\
i \ar@/^0.7pc/[rr]^{} & &i' \ar@/^0.7pc/[ll]^{}\\
}$
\caption{\label{fig:3} A directed graph}
\end{figure}
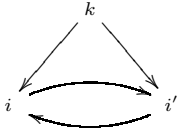

So the game $G'$ is not exactly the social network game
associated with the weighted directed graph related to $i \to i'\to
i$. However, the observation stated after the proof of
Theorem~\ref{thm:2-FIP} allows us to conclude that $G'$ does have the
FIP.
\end{proof}

\subsection{Graphs with no source nodes}
We conclude with the following analogue of
Theorem~\ref{thm:fbrp-hard-nosource}.

\begin{theorem} \label{thm:fip-hard-nosource}
For a network $\snet$ whose underlying graph has no source nodes,
deciding whether the game $\mathcal{G}(\snet)$ has the FIP is
co-NP-hard.
\end{theorem}

\begin{proof}
The proof is a modification of Theorem~\ref{thm:fip-hard}.  Given an
instance of the PARTITION problem we use the following modification of
the network given in Figure~\ref{fig:partition-fip}.  We `twin' each
node $i \in \{1, \LL, n\}$ with a new node $i'$, also with the product
set $\{t_1, t_2\}$, by adding the edges $(i,i')$ and $(i',i)$. We also
`twin' the node marked $\{t_3\}$, call it $d$, with a new node $d'$
also with the product set $\{t_3\}$, by adding the edges $(d,d')$ and
$(d',d)$.  Next, we choose the weights $w_{i i'}, w_{i' i}$, where $i
\in \{1, \LL, n\}$, and $w_{d d'}$ and $w_{d' d}$ and the
corresponding thresholds so that when $i$ and $i'$ or $d$ and $d'$
adopt a common product, their payoffs are positive.

Suppose that a solution to the considered instance of the PARTITION
problem exists. Then we extend the joint strategy considered in the
proof of Theorem~\ref{thm:fip-hard} by additionally assigning $t_1$ to
each node $i'$ such that $i \in S$, $t_2$ to each node $i'$ such that
$i \in \{1, \LL, n\} \setminus S$ and $t_3$ to nodes $d$ and $d'$.
Then, as before, there is no finite improvement path starting in this
joint strategy, so $\mathcal{G}(\snet)$ does not have the FIP.

Suppose that no solution to the considered instance of the PARTITION
problem exists. We show that then the game has the FIP. Consider an
improvement path $\xi$. By Theorem~\ref{thm:2-FIP}, each network
consisting of just the twined nodes $(1,1'), \ldots, (n,n')$ and
$(d,d')$ has the FIP. Therefore, there exists a joint strategy
$\strprofile^k$ in $\xi$ such that in the suffix of $\xi$ starting at
$\strprofile^k$, none of the following nodes are scheduled:
$1,1',\ldots,n,n',d,d'$. Again, we have the two cases as analysed in
Theorem~\ref{thm:fip-hard}. When $\val(t_2) < \frac14$, using the same
argument as in Theorem~\ref{thm:fip-hard} we conclude that the
improvement path $\xi$ is finite. In the case when $\val(t_2) >
\frac14$ the argument in Theorem~\ref{thm:fip-hard} shows that there
exists an index $k_b >k$ in $\xi$ such that for the suffix of $\xi$
starting at $k_b$ (call it $\xi_b$), node $b$ never chooses $t_1$. If
$d$ and $d'$ choose $t_0$ in $\strprofile^k$, since $b$ never chooses
$t_1$ in $\xi_b$, the better response of $c$ in $\xi_b$ is $t_0$. From
this it follows that $\xi_b$ and hence $\xi$ is finite. If $d$ and
$d'$ choose $t_3$ in $\strprofile^k$, then it follows by the argument
given in Theorem~\ref{thm:fip-hard} that $\xi$ is finite.  
\end{proof}

\section{The uniform FIP property: general case}
\label{sec:uniform-FIP}

%% Let us reconsider the games associated with the networks from
%% Figures~\ref{fig:br0}(a) and \ref{fig:br1}(a).  In spite of the fact
%% that they do not have the FIP, for any initial joint strategy there
%% exists a finite improvement path.
Let us reconsider the games associated with the networks from
Figures~\ref{fig:br0}(a) and \ref{fig:br1}(a). Since these games do
not have the FBRP, by definition, they do not have the FIP either. In
spite of this fact, for any initial joint strategy there exists a
finite improvement path.  Indeed, it suffices to schedule the players
in the clockwise manner, that is to choose each time the first player
who did not select a best response.  This is an instance of a more
general result proved in the next section.

By a \bfe{scheduler} we mean a function $f$ that given a joint
strategy $s$ that is not a Nash equilibrium selects a player who did
not select a best response in $s$. An improvement path
$\xi= \strprofile^1,\strprofile^2,\ldots$ \bfe{respects} a
scheduler $f$ if for all $k$ smaller than the length of $\xi$,
$\strprofile^{k+1}=(\strprofile_i',\strprofile^{k}_{-i})$, where
$f(\strprofile^k)=i$. We say that a strategic game has the
\bfe{uniform FIP} if there exists a scheduler $f$ such that all
improvement paths $\rho$ which respect $f$ are finite.
%% The property
%% of having the uniform FIP is stronger than that of being \bfe{weakly
%% acyclic} (see \cite{You93} and \cite{Mil96}) that only guarantees that for any strategy
%% profile there exists a finite improvement path that starts at it.
%% In fact consider the following example.

In this and the next section we investigate the complexity of
determining whether a social network game has the uniform FIP.  We
begin by establishing the following counterpart of Theorem
\ref{thm:fip-hard}.

\begin{theorem} \label{thm:ufip-hard}
Deciding whether for a network $\snet$ the game
$\mathcal{G}(\snet)$ has the uniform FIP is co-NP-hard.
\end{theorem}

\begin{proof}
  We prove that the complement of the problem is NP-hard.  As before
  we use an instance of the PARTITION problem in the form of $n$
  positive rational numbers $(a_1,\LL,a_n)$ such that $\sum_{i=1}^n
  a_i = 1$.  To construct the appropriate network $\snet$, as in the
  proof of Theorem~\ref{thm:np}, we employ the networks given in
  Figures~\ref{fig:noNe1} and \ref{fig:partition}. However, we use
  different product sets in the second network and combine these two
  networks differently.
  
  For each node $i\in\{1,\LL,n\}$ we set $P(i) = \{t_1, t_2\}$.
  Further we set $P(a) = \{t_1\}$ and $P(b) = \{t_2\}$.  As before we
  set $w_{i a} = w_{i b} = a_i$ and we assume that the threshold of
  the nodes $a$ and $b$ is constant and equals $\frac12$.  We now
  identify the nodes $a$ and $b$ of the network from
  Figure~\ref{fig:partition} respectively with the nodes marked by
  $\{t_1\}$ and $\{t_2\}$ in the network in Figure~\ref{fig:noNe1}.
  The resulting network is given in Figure~\ref{fig:ufip-hard}.

\begin{figure}[ht]
\centering
$
\def\objectstyle{\scriptstyle}
\def\labelstyle{\scriptstyle}
\xymatrix@W=10pt @R=27pt @C=15pt{
1 \ar[d]_{} \ar[rrrrd]_{}& 2 \ar[ld]^{} \ar@/^0.7pc/[rrrd]^{}& & \cdots &n \ar[lllld]^{} \ar[d]^{} \\
a \ar@{}[u]^<{\{t_1\}} \ar[d]{} & & & &b \ar@{}[u]_<{\{t_2\}} \ar[d]{} \\
c  \ar[rrd]^{} \ar@{}[u]^<{\{t_1,t_2\}}& & & &d  \ar[llll]^{} \ar@{}[u]_<{\{t_2,t_3\}}\\
& &e \ar[rru]^{} \ar@{}[rru]_<{\{t_1,t_3\}}\\
& &g \ar[u]^{} \ar@{}[u]_<{\{t_3\}}\\
}$
\caption{\label{fig:ufip-hard}A network related to the uniform FIP}
\end{figure}
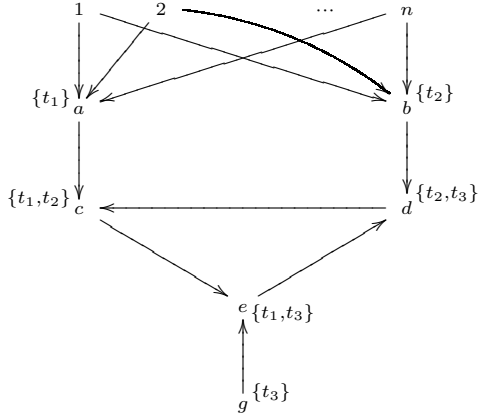

Suppose now that a solution to the considered instance of the
PARTITION problem exists, that is for some set $S \sse \{1, \LL, n\}$
we have $\sum_{i\in S} a_i = \sum_{i\not\in S} a_i = \frac{1}{2}$.
Consider the game $\mathcal{G}(\snet)$. 
Take the joint strategy $s$ formed by the following strategies:

\begin{itemize}

\item $t_1$ assigned to each node $i \in S$ and the nodes $a$ and $c$,

\item $t_2$ assigned to each node $i \in \{1, \LL, n\} \setminus S$
and the nodes $b$ and $d$,

\item $t_3$ assigned to the nodes $e$ and $g$.

\end{itemize}

Any improvement path that starts in this joint strategy will not
change the strategies assigned to the nodes $a, b$ and $g$. So if such
an improvement path terminates, it produces a Nash equilibrium in the
game associated with the network given in Figure~\ref{fig:noNe1} of
Example~\ref{exa:nonash}. But we showed there that this
game does not have a Nash equilibrium. Consequently, there is no
finite improvement path in the game $\mathcal{G}(\snet)$ that starts
in the above joint strategy and a fortiori $\mathcal{G}(\snet)$ does
not have the uniform FIP.

Suppose now that no solution to the considered instance of the
PARTITION problem exists. We show that then
the game $\mathcal{G}(\snet)$ has the uniform FIP.
To this end we order the nodes of $\snet$ as follows
(note the positions of the nodes $c, d$ and $e$):
\[
1, 2, \LL, n, a, b, g, c, e, d.
\]
Consider the scheduler $f$ that given a joint strategy selects the
first node in the above list that did not select a best response. 
(An aside: a node can be selected more than once --
such a possibility is discussed in the second case below.)

Take an improvement path that respects this scheduler.  After at most
$n$ steps the nodes $1, 2, \LL, n$ all selected a product $t_1$ or
$t_2$.  Suppose first that $ \sum_{i \in \{1, \LL, n\} \mid s_i = t_1}
w_{i a} > \frac12 $.

In the considered improvement path eventually
the following selections have been made
by the nodes $a, b, g, c$ and $e$:
\[
a: t_1, b: t_0, g: t_3, c: t_1, e: t_1.
\]
Either the node $d$ initially selected $t_0$ or it 
is eventually scheduled and then it switches to $t_0$.
At this moment a Nash equilibrium is reached, that is, the 
considered improvement path is finite.

Suppose next that
$
\sum_{i \in \{1, \LL, n\} \mid s_i = t_2} w_{i b} > \frac12
$.
Then after a number of steps the following selections have been made by the nodes
$a, b$ and $g$:
\[
a: t_0, b: t_2, g: t_3.
\]

Either the node $c$ initially selected $t_0$ or $t_2$, or it is
eventually scheduled and then it switches to $t_0$ or $t_2$.  Further,
either the node $e$ initially selected $t_3$ or it is eventually
scheduled and then it switches to $t_3$.  Finally, if the node $d$ is
eventually scheduled, then it switches to $t_3$, as well.  If now the
node $c$ is scheduled again, then it switches to $t_0$.  At this
moment a Nash equilibrium is reached and the considered improvement
path terminates.
\end{proof}

\section{The uniform FIP: special cases}
\label{sec:uniform-FIP-special}

As in the case of the Nash equilibria and the FIP we now consider
the uniform FIP for the special cases of social network games.

\subsection{Simple cycles}

In the case when the underlying graph is a simple cycle the following
positive result holds.

\begin{theorem}
\label{thm:cycle-UFIP}
Let $\snet$ be a network
such that the underlying graph is a simple cycle.
Then the game $\mathcal{G}(\snet)$ has the uniform FIP.
\end{theorem}

\begin{proof}
We use the scheduler $f$ that given a joint strategy
$s$ chooses the smallest index $i$ such that
$s_i$ is not a best response to $s_{-i}$.  So this scheduler
selects a player again if he did not switch to a best response.
Therefore we can assume that each selected player immediately selects
a best response.

Consider a joint strategy $s$ taken from a best response improvement
path.  Observe that for all $k$ if $s_k \in P(k)$ and $p_k(s) \geq 0$
(so in particular if $s_k$ is a best response to $s_{-k}$), then $s_k
= s_{k \ominus 1}$. Consequently the following property holds for all
$i > 1$:
\[
\mbox{$Z(i)$: if $f(s) = i$ and $s_{i-1} \in P(i-1)$ then for all $j \in
\{n,1,\LL, i-1\}$, $s_j = s_{i-1}$.}
\]
In words: if $i$ is the first player who did not choose a best
response and player $i-1$ strategy is a product, then this product
is the strategy of every earlier player and of player $n$.
  
Along each best response improvement path that respects $f$ the
value of $f(s)$ strictly increases until the path terminates or at
certain stage $f(s) = n$. In the latter case if $s_{n-1} = t_0$, then
the unique best response for player $n$ is $t_0$. Otherwise $s_{n-1}
\in P(n-1)$, so on the account of property $Z(n)$ all players'
strategies equal the same product and the payoff of player $n$ is
negative (since $f(s) = n$).  So the unique best response for player
$n$ is $t_0$, as well.

This switch begins a new round with player 1 as the next scheduled player.
Player 1 also switches to $t_0$ and from now on every consecutive
player switches to $t_0$, as well. The resulting path terminates once
player $n-2$ switches to $t_0$.
\end{proof}

Another scheduler $f$ that we could use is the one that given $s$
chooses the smallest
index $i$ such that $p_i(s) < 0$. Note that if $p_i(s) < 0$, then
$i$ does not play a best response in $s$. The argument is the same as for
the above scheduler. 

These two schedulers differ. Take for instance the network given in
Figure~\ref{fig:br1} and the initial joint strategy
$(t_2,t_2,t_1)$. The following improvement path terminating in the
trivial Nash equilibrium respects the second scheduler:

\[(\underline{t_2},t_2,t_1) \Rightarrow
(t_0,\underline{t_2},t_1) \Rightarrow (t_0,t_0,\underline{t_1})
\Rightarrow (t_0,t_0,t_0).\]

However, using the first scheduler, node 1 is scheduled again in the
joint strategy $(t_0,t_2,t_1)$ since $t_0$ is not its best response
and similarly for node 2. Thus the improvement path terminates in the
Nash equilibrium $(t_1,t_1,t_1)$.

\subsection{Graphs with no source nodes}

Finally, we consider social network games in the case when
the underlying graph has no source nodes. Then it is
possible that for some initial joint strategy \emph{no} improvement
path starting at it terminates, so in particular the game does not have the 
uniform FIP. To see this consider the network
given in Figure~\ref{fig:br-counterex1}(a).

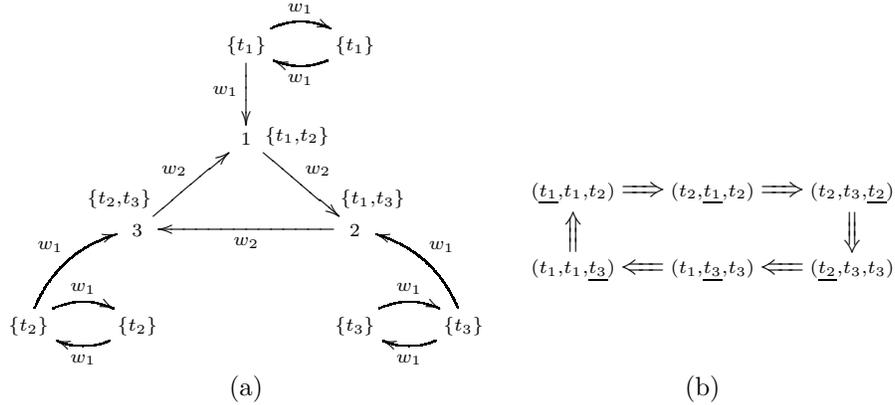
\begin{figure}[ht]
\centering
\begin{tabular}{ccc}
$
\def\objectstyle{\scriptstyle}
\def\labelstyle{\scriptstyle}
\xymatrix@W=10pt @C=20pt{
& & \{t_1\} \ar[d]_{w_1} \ar@/^0.7pc/[r]^{w_1}& \{t_1\} \ar@/^0.7pc/[l]^{w_1}\\
& &1 \ar[rd]^{w_2} \ar@{}[rd]^<{\{t_1,t_2\}}\\
 &3 \ar[ur]^{w_2} \ar@{}[ur]^<{\{t_2,t_3\}}& &2 \ar[ll]^{w_2} \ar@{}[lu]_<{\{t_1,t_3\}} \\
\{t_2\} \ar@/^0.7pc/[ru]^{w_1} \ar@/^0.7pc/[r]^{w_1}& \{t_2\} \ar@/^0.7pc/[l]^{w_1}& & \{t_3\} \ar@/^0.7pc/[r]^{w_1}&\{t_3\} \ar@/_0.7pc/[lu]_{w_1} \ar@/^0.7pc/[l]^{w_1}
%% \{t_2\} \ar[ru]^{w_1} \ar@/^0.7pc/[d]^{w_1}& & & &\{t_3\} \ar[lu]_{w_1} \ar@/^0.7pc/[d]^{w_1}\\
%% \{t_2\} \ar@/^0.7pc/[u]^{w_1}& & & & \{t_3\} \ar@/^0.7pc/[u]^{w_1}
}$
&
%&
\raisebox{-2.5cm}{\parbox{8cm}{
$
\def\objectstyle{\scriptstyle}
\def\labelstyle{\scriptstyle}
\xymatrix@W=10pt @C=15pt @R=15pt{
(\underline{t_1},t_1,t_2)\ar@{=>}[r]& (t_2,\underline{t_1},t_2)\ar@{=>}[r]& (t_2,t_3,\underline{t_2})\ar@{=>}[d]\\
(t_1,t_1,\underline{t_3})\ar@{=>}[u]& (t_1,\underline{t_3},t_3)\ar@{=>}[l]& (\underline{t_2},t_3,t_3)\ar@{=>}[l]\\
}$
}}
\\
(a) & \parbox{3.7cm}{(b)}\\
\end{tabular}
\caption{\label{fig:br-counterex1}Another network with an infinite improvement path}
\end{figure}

We assume that each threshold is a constant $\theta$, where $\theta <
w_1 < w_2$. Consider the joint strategy $\strprofile$, in which the
nodes marked by $\{t_1\}$, $\{t_2\}$ and $\{t_3\}$ choose the unique
product in their product sets and players 1, 2 and 3 choose $t_1, t_1$
and $t_2$, respectively. For convenience, we denote $\strprofile$ by
$(t_1,t_1,t_2)$. This joint strategy is not a Nash equilibrium since
player 1 gains by deviating to $t_2$. There is a unique improvement
path starting at $\strprofile$ and this path is infinite.  It is shown
in Figure~\ref{fig:br-counterex1}(b). In each joint strategy, we
underline the strategy that is not a best response.

By using a more complex network that builds upon the 
above one we now prove the following analogue of
Theorem~\ref{thm:ufip-hard}.

\begin{theorem} \label{thm:ufip-hard2}
For a network $\snet$ whose underlying graph has no source
nodes, deciding whether the game $\mathcal{G}(\snet)$ has 
the uniform FIP is co-NP-hard.
\end{theorem}

\begin{proof}
The proof extends the proof of Theorem~\ref{thm:ufip-hard}.  Given
an instance of the PARTITION problem we use the following modification
of the network given in Figure~\ref{fig:ufip-hard}.  We `twin' each
node $i \in \{1, \LL, n\}$ with a new node $i'$, also with the product
set $\{t_1, t_2\}$, by adding the edges $(i,i')$ and $(i',i)$. We also
`twin' the node $g$ with a new node $g'$, also with the product set
$\{t_3\}$, by adding the edges $(g,g')$ and $(g',g)$.  Next, we choose
the weights $w_{i i'}, w_{i' i}$, where $i \in \{1, \LL, n\}$, and
$w_{g g'}$ and $w_{g' g}$ and the corresponding thresholds so that
when $i$ and $i'$ or $g$ and $g'$ adopt a common product, their payoff
is positive.

Suppose that a solution to the considered instance of the
PARTITION problem exists. Then we extend the joint strategy considered
in the proof of Theorem~\ref{thm:ufip-hard} by additionally
assigning $t_1$ to each node $i'$ such that $i \in S$,
$t_2$ to each node $i'$ such that $i \in \{1, \LL, n\} \setminus S$
and $t_3$ to the node $g'$.
Then, as before, there is no finite improvement path starting
in this joint strategy, so  $\mathcal{G}(\snet)$ does not have uniform FIP.

Suppose now that no solution to the considered instance of the
PARTITION problem exists. We show that then the game $\mathcal{G}(\snet)$ has the uniform FIP.
To this end we now use the following order of the nodes of $\snet$:
\[
1, 1', 2, 2', \LL, n, n', g, g', a, b, c, e, d,
\]
and as before use the scheduler that `follows' this ordering.
(The reader may notice that the primed nodes will never be scheduled.)

When considering the possible joint strategies in
$\mathcal{G}(\snet)$ we now need to consider new cases when
after being scheduled some of
the nodes $i \in \{1, \LL, n\}$ or $g$ end up selecting the strategy
$t_0$.  In each case the reasoning is similar as in the proof of
Theorem~\ref{thm:ufip-hard} and allows us to conclude that each
improvement path terminates with a Nash equilibrium. We only mention
the final selection of the strategies by the nodes $a, b, c, d, e$
and $g$ and leave checking the details to the reader.  
\II

\NI
\emph{Case 1}.
$
\sum_{i \in \{1, \LL, n\} \mid s_i = t_1} w_{i a} < \frac12
$, 
$
\sum_{i \in \{1, \LL, n\} \mid s_i = t_2} w_{i b} < \frac12
$ and $g$ selects $t_3$.

The final selection of the strategies by the nodes $a, b, c, d$ and $e$ is then
\[
a: t_0, b: t_0, c: t_0, d: t_3, e: t_3.
\]

\NI
\emph{Case 2}.
$
\sum_{i \in \{1, \LL, n\} \mid s_i = t_1} w_{i a} < \frac12
$, 
$
\sum_{i \in \{1, \LL, n\} \mid s_i = t_2} w_{i b} < \frac12
$ and $g$ selects $t_0$.

The final selection of the strategies by the nodes $a, b, c, d$ and $e$ is then
\[
a: t_0, b: t_0, c: t_0, d: t_0, e: t_0.
\]

\NI
\emph{Case 3}.
$
\sum_{i \in \{1, \LL, n\} \mid s_i = t_1} w_{i a} \geq \frac12
$, 
$
\sum_{i \in \{1, \LL, n\} \mid s_i = t_2} w_{i b} < \frac12
$ and $g$ selects $t_3$.

The final selection of the strategies by the nodes $a, b, c, d$ and $e$ is then
\[
\mbox{
$a: t_1, b: t_0, c: t_1, d: t_0, e: t_1$
or 
$a: t_0, b: t_0, c: t_0, d: t_3, e: t_3$,
}
\]
% \[
% a: t_1, b: t_0, c: t_1, d: t_0, e: t_1
% \]
% or
% \[
% a: t_0, b: t_0, c: t_0, d: t_3, e: t_3,
% \]
(where the latter can only arise when $\sum_{i \in \{1, \LL, n\} \mid s_i = t_1} w_{i a} = \frac12$).
\II

\NI
\emph{Case 4}.
$
\sum_{i \in \{1, \LL, n\} \mid s_i = t_1} w_{i a} \geq \frac12
$, 
$
\sum_{i \in \{1, \LL, n\} \mid s_i = t_2} w_{i b} < \frac12
$ and $g$ selects $t_0$.

The final selection of the strategies by the nodes $a, b, c, d$ and $e$ is then
\[
\mbox{
$a: t_1, b: t_0, c: t_1, d: t_0, e: t_1$ or $a: t_0, b: t_0, c: t_0, d: t_0, e: t_0$,
}
\]
% \[
% a: t_1, b: t_0, c: t_1, d: t_0, e: t_1
% \]
% or
% \[
% a: t_0, b: t_0, c: t_0, d: t_0, e: t_0,
% \]
(where the latter can only arise when $\sum_{i \in \{1, \LL, n\} \mid s_i = t_1} w_{i a} = \frac12$).
\II

\NI
\emph{Case 5}.
$
\sum_{i \in \{1, \LL, n\} \mid s_i = t_1} w_{i a} < \frac12
$, 
$
\sum_{i \in \{1, \LL, n\} \mid s_i = t_2} w_{i b} \geq \frac12
$ and $g$ selects $t_3$.

The final selection of the strategies by the nodes $a, b, c, d$ and $e$ is then
\[
\mbox{
$a: t_0, b: t_2, c: t_0, d: t_3, e: t_3$ or $a: t_0, b: t_0, c: t_0, d: t_3, e: t_3$,
}
\]
% \[
% a: t_0, b: t_2, c: t_0, d: t_3, e: t_3
% \]
% or
% \[
% a: t_0, b: t_0, c: t_0, d: t_3, e: t_3,
% \]
(where the latter can only arise when $\sum_{i \in \{1, \LL, n\} \mid s_i = t_2} w_{i b} = \frac12$).
\II

\NI
\emph{Case 6}.
$
\sum_{i \in \{1, \LL, n\} \mid s_i = t_1} w_{i a} < \frac12
$, 
$
\sum_{i \in \{1, \LL, n\} \mid s_i = t_2} w_{i b} \geq \frac12
$ and $g$ selects $t_0$.

The final selection of the strategies by the nodes $a, b, c, d$ and $e$ is then
\[
\mbox{
$a: t_0, b: t_2, c: t_2, d: t_2, e: t_0$ or $a: t_0, b: t_0, c: t_0, d: t_0, e: t_0$,
}
\]
% \[
% a: t_0, b: t_2, c: t_2, d: t_2, e: t_0
% \]
% or
% \[
% a: t_0, b: t_0, c: t_0, d: t_0, e: t_0,
% \]
(where the latter can only arise when $\sum_{i \in \{1, \LL, n\} \mid s_i = t_2} w_{i b} = \frac12$).
\II

This proves that
$\mathcal{G}(\snet)$ has the uniform FIP.
\end{proof}

Analogously to the notion of a game having the uniform FIP one can
introduce the notion of a game having the \bfe{uniform FBRP}, by
simply restricting attention to the best response improvement paths.
It is easy to check that all the proofs given in this section show
that the obtained results also hold when we replace the reference to
the uniform FIP by that to the uniform FBRP. In fact, the latter
property admits less improvement paths, so some proofs become simpler.
We shall return to this matter in the last section.

\section{Weakly acyclic games}
\label{sec:WA}

Finally, we show that for the social network games the property of
being weakly acyclic is weaker than that of having the uniform FIP.

\begin{example}
\rm Consider the network $\snet$ given in Figure
\ref{exa:weakly}(a). This is essentially the network in Figure
\ref{fig:noNe1} except for the addition of the new source node with
the product set $\{t_4\}$ as a neighbour of node 1 and the new product
$t_4$ included in the product sets of players 1,2 and 3. We assume that
each threshold is a constant $\theta$, where $\theta < w_3 < w_1 <
w_2$.

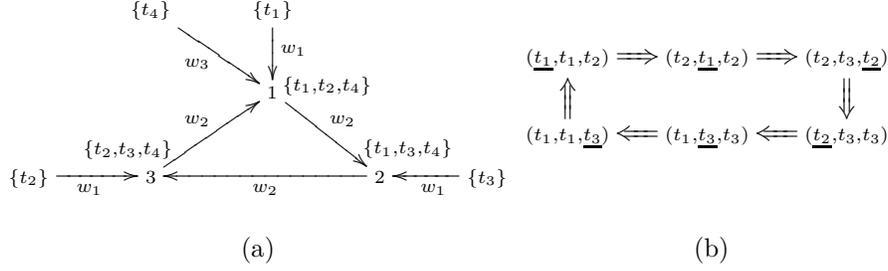
\begin{figure}[ht]
\centering
\begin{tabular}{ccc}
$
\def\objectstyle{\scriptstyle}
\def\labelstyle{\scriptstyle}
\xymatrix@R=20pt @C=25pt{
& \{t_4\} \ar[dr]_{w_3} &\{t_1\} \ar[d]^{w_1}\\
& &1 \ar[rd]^{w_2} \ar@{}[rd]^<{\{t_1,t_2,t_4\}}\\
\{t_2\} \ar[r]_{w_1} &3 \ar[ur]^{w_2} \ar@{}[ur]^<{\{t_2,t_3,t_4\}}& &2 \ar[ll]^{w_2} \ar@{}[lu]_<{\{t_1,t_3,t_4\}} &\{t_3\} \ar[l]^{w_1}\\
%%\{t_2\} \ar[ru]_{w_1} & & & &\{t_3\} \ar[lu]^{w_1}\\
}$
&\hspace{-0.5cm}
%&
\raisebox{-1.2cm}{\parbox{8cm}{
$
\def\objectstyle{\scriptstyle}
\def\labelstyle{\scriptstyle}
\xymatrix@W=10pt @C=15pt @R=15pt{
(\underline{t_1},t_1,t_2)\ar@{=>}[r]& (t_2,\underline{t_1},t_2)\ar@{=>}[r]& (t_2,t_3,\underline{t_2})\ar@{=>}[d]\\
(t_1,t_1,\underline{t_3})\ar@{=>}[u]& (t_1,\underline{t_3},t_3)\ar@{=>}[l]& (\underline{t_2},t_3,t_3)\ar@{=>}[l]\\
}$
}}
\\\\
(a) & \parbox{3.7cm}{(b)}\\
\end{tabular}
\caption{\label{exa:weakly} A weakly acyclic game that does not have the
  uniform FIP}
\end{figure}

We first show that the game $\mathcal{G}(\snet)$ does not have the
uniform FIP. Assume that each source node selects its unique
product. Identify each joint strategy that extends this selection with
the selection of the strategies by the players 1, 2 and 3. Consider
the infinite improvement path starting at the joint strategy
$s=(t_1,t_1,t_2)$ shown in Figure \ref{exa:weakly}(b). In each joint
strategy in this improvement path, there is a unique player who is not
playing his best response. Therefore, irrespective of the scheduler
chosen, there is always a unique player who can be scheduled. It
follows that $\mathcal{G}(\snet)$ does not have the uniform FIP.

On the other hand, this game is weakly acyclic. It suffices to
consider joint strategies in which each source node selects its
unique product. First note that any improvement path in which player 1
never adopts product $t_2$ is finite since he cannot switch from $t_1$
to $t_4$. We have the following cases, where as above, we only list
the strategies of players 1, 2 and 3.

\II \NI \emph{Case 1.} $\strprofile =(t_2,x,y)$ where $y \neq
t_2$. Then by definition $\payoff_1(\strprofile)<0$. If $x \neq t_4$
then consider the improvement path $\strprofile=(\underline{t_2},x,y)
\Rightarrow (t_4,\underline{x},y) \Rightarrow (t_4,t_4,\underline{y})
\Rightarrow (t_4,t_4,t_4)$. So player 1 first switches to $t_4$ (which
is a better response) instead of switching to $t_1$ (which is the best
response) and subsequently players 2 and 3 switch to $t_4$. This
results in the joint strategy $(t_4,t_4,t_4)$ which is a Nash
equilibrium. If $x = t_4$ then we have a shorter improvement path
$(t_4,t_4,\underline{y}) \Rightarrow (t_4,t_4,t_4)$.

\II \NI \emph{Case 2.} $\strprofile=(t_2,x,t_2)$. If $x \neq t_3$ then
$(t_2, \underline{x},t_2) \Rightarrow (t_2,t_3,\underline{t_2})
\Rightarrow (t_2,t_3,t_3)$. If $x=t_3$ then $(t_2,x,\underline{t_2})
\Rightarrow (t_2,t_3,t_3)$. But the proof of Case 1 shows that
$(t_2,t_3,t_3) \Rightarrow^* (t_4,t_4,t_4)$, so irrespective of the
value of $x$ there is a finite improvement path starting in
$\strprofile$.
\HB
\end{example}

In turn, there are social network games which have Nash equilibria but
are not weakly acyclic. 

\begin{example}
\rm

Consider the modification of the network in
Figure \ref{exa:weakly} where we `twin' the node with product set
$\{t_4\}$, call it 4, with a new node 5 also with the product set
$\{t_4\}$ by adding edges $(4,5)$ and $(5,4)$. Let the weight of each
new edge be a constant $w \geq \theta$. From the previous example, it
is easy to see that the joint strategy $(t_4,t_4,t_4,t_4,t_4)$ (where
we list the choice of nodes 1 to 5) is a Nash equilibrium. On the
other hand, consider the joint strategy $\strprofile=(t_1, t_1,
t_2,t_0,t_0)$. It corresponds to the joint strategy $(t_1,t_1,t_2)$
analysed in Example \ref{exa:nonash}, where we showed that every
improvement path starting at it is infinite. The same holds here and
therefore the game is not weakly acyclic.
\HB
\end{example}
\begin{theorem} \label{thm:ufip-hard3}
For an arbitrary network (respectively, a network whose underlying
graph has no source nodes) $\snet$, deciding whether the game
$\mathcal{G}(\snet)$ is weakly acyclic is co-NP-hard.
\end{theorem}

\begin{proof}
Consider the case of an arbitrary
network. In the proof of Theorem~\ref{thm:ufip-hard} we showed that if a
solution to the instance of the PARTITION problem exists then in the
associated game $\mathcal{G}(\snet)$, there is a joint strategy
$\strprofile$ such that all improvement paths starting in
$\strprofile$ are infinite. This implies that $\mathcal{G}(\snet)$ is
not weakly acyclic.

Conversely, if there is no solution to the instance of the PARTITION
problem then there is a scheduler $f$ such that all improvement paths
in $\mathcal{G}(\snet)$ that respects $f$ are finite. This implies
that $\mathcal{G}(\snet)$ is weakly acyclic.

The proof for the case of a network whose underlying
graph has no source nodes is analogous and uses the proof of 
Theorem~\ref{thm:ufip-hard2}.
\end{proof}

\section{Concluding remarks}
\label{sec:conc}

\subsection{Summary of the results}

In this paper we studied various aspects of product adoption by agents
who form a social network by focussing on a natural class of strategic
games associated with the class of social networks introduced in
\cite{AM11}. We identified three natural types of (pure) Nash
equilibria in these games: arbitrary, non-trivial, and determined, and
focussed on games associated with four classes of social networks:
arbitrary ones and those whose underlying graph is a DAG, or a simple
cycle, or has no source nodes.  We also showed that the price of
anarchy and the price of stability is unbounded, even if we limit
ourselves to the social networks whose underlying graph is a DAG or
has no source nodes.

Further, we studied the finite best response property (FBRP), the
finite improvement property (FIP) and also introduced a new class of
 games that have the uniform FIP. The following table summarizes
our complexity and existence results, where we refer to the underlying
graph having $n$ nodes.

\begin{center}
        \begin{tabular}{|l|c|c|c|c|}
\hline
~~~~~property & arbitrary   & DAG              & simple cycle                       & no source \\
               &             &                  &                                    &nodes \\
\hline
Arbitrary NE   & NP-complete & always exists    & always exists                      & always exists \\
Non-trivial NE & NP-complete & always exists    & $\bigo(|\products| \cdot n)$       & $\bigo(|\products| \cdot n^3)$ \\
Determined NE  & NP-complete & NP-complete      & $\bigo(|\products| \cdot n)$       & NP-complete \\
FBRP           & co-NP-hard  & yes              & $\bigo(|\products| \cdot n)$       & co-NP-hard \\ 
FIP            & co-NP-hard  & yes              & ?                                  & co-NP-hard \\ 
Uniform FIP    & co-NP-hard  & yes              & yes                                & co-NP-hard\\
Weakly acyclic & co-NP-hard  & yes              & yes                                & co-NP-hard\\
\hline
\end{tabular}
\end{center}

We also mentioned in the Section \ref{sec:uniform-FIP-special} the
notion of a game having the uniform FBRP.  In a recent paper,
\cite{AS12}, we developed the notion of schedulers further and in
particular showed that for finite games the notions of having the
uniform FIP and the uniform FBRP are equivalent.  This provides
another way of showing that the results of the previous section also
hold for the latter notion.

\subsection{Final comments}

In the definition of the social network games we took a number of
simplifying assumptions. In particular, we stipulated that the source
nodes have a constant payoff $\constutil > 0$.  One could allow the
source nodes to have arbitrary positive utility for different
products. This would not affect the proofs. Indeed, in Nash
equilibria the source nodes would select only the products with the
highest payoff, so the other products in their product sets could be
disregarded. Further, the FBRP, the FIP, the uniform FIP and weak
acyclicity of a social network game is obviously not affected by such
a modification.

We could also allow the weights to be parametrized by a product. The
corresponding expression in the definition of the payoff function
would then become $\sum_{j \in \inflset_i^t(\strprofile)}
w_{ji}(t) - \theta(i,t)$. In contrast, as shown in Example~\ref{exa:12},
such a modification can affect some of the positive results.

Further, the results of this paper can be slightly generalized by using a more
general notion of a threshold that would also depend on the set
of neighbours who adopted a given product. In this more general setup
for $i \in V$, $t \in \prodset(i)$ and $X \subseteq
\neighbour(i)$, the \bfe{threshold function} $\theta$ yields a 
value $\theta(i,t,X) \in (0,1]$. 

For the results to continue to hold one needs to 
assume that the threshold function satisfies the following
\bfe{monotonicity} condition: 
if $X_1 \subseteq X_2$ then $\theta(i,t,X_1) \geq \theta(i,t,X_2)$.
Intuitively, agent $i$'s resistance to adopt a product
decreases when the set of its neighbours who adopted it increases.
We decided not to use this definition for the sake of readability.

This work can be pursued in a couple of natural directions. One is the
study of social networks with other classes of underlying graphs.
Another is an investigation of the complexity results for other
classes of social networks, in particular for the equitable ones,
i.e., networks in which the weight functions are defined as
$\weight_{ji} = \frac{1}{|\neighbour(i)|}$ for nodes $i$ and $j \in
\neighbour(i)$. One could also consider other equilibrium concepts
like strict Nash equilibria.

Currently we started a study of slightly different games, in which the
players are obliged to choose a product, i.e., games in which the
strategy $t_0$ is absent. Such games naturally correspond to
situations in which the agents always choose a product, for instance a
subscription for their mobile telephone. These games substantially
differ from the ones considered here. For example, Nash equilibria
do not need to exist when the underlying graph is a simple cycle.

Finally, we initiated in \cite{AMS13} an analysis of the consequences
of introducing new products in a social network.  This can be done
using the framework of the social network games by focusing on the
finite best response property and the finite improvement property.  In
particular we found that in some cases such a product addition can
lead to a different product selection in which, remarkably, each agent
is strictly worse off.

\section*{Acknowledgments}

We would like to thank the referee for an exceptionally detailed and
useful report. Theorem~\ref{thm:polymatrix-NP} was suggested by
Bernhard von Stengel.

\bibliographystyle{abbrv}
\bibliography{/ufs/apt/bib/e.bib}
%\bibliography{e.bib} 

\end{document}